\documentclass[12pt]{article}

\usepackage{amsmath, amsthm, amssymb, amsfonts, graphicx, mathtools, color, setspace, fancyhdr}
\usepackage[margin=1in]{geometry}
\usepackage[utf8]{inputenc}
\usepackage[english]{babel}
\usepackage[svgnames]{xcolor}
\usepackage{tikz}
\usetikzlibrary{positioning, patterns}
\usepackage{setspace}
\usepackage{caption}
\usepackage[labelformat=simple]{subcaption}

\usepackage{epsfig}
\usepackage[numbers]{natbib}
\usepackage[linkcolor=Blue,colorlinks=true,citecolor=ForestGreen,filecolor=red,urlcolor=Blue]{hyperref}
\usepackage{cleveref}[2011/12/24]
\usepackage[pagewise,displaymath,mathlines]{lineno}
\usepackage{authblk}
\usepackage{enumerate}

\usepackage{thmtools, thm-restate}

\newtheorem{theorem}{Theorem}

\newtheorem{lemma}[theorem]{Lemma}
\newtheorem{claim}[theorem]{Claim}
\newtheorem{observation}[theorem]{Observation}

\theoremstyle{definition}

\newcommand{\tw}{\mathrm{tw}}
\newcommand{\cutsim}{\mathrm{cutsim}}
\newcommand{\simw}{\mathrm{simw}}
\newcommand{\cutmim}{\mathrm{cutmim}}
\newcommand{\mimw}{\mathrm{mimw}}
\newcommand{\ssi}{\subseteq_i}
\newcommand{\tin}{\mathrm{tree}\textnormal{-}\alpha}

\renewcommand{\ssi}{\subseteq_i}

\newcommand{\comment}[1]{}

\newcommand*{\myproofname}{Proof}
\newenvironment{claimproof}[1][\myproofname]{\begin{proof}[#1]}{\end{proof}}

\providecommand{\keywords}[1]
{
  \small	
  \textbf{\textit{Keywords---}} #1
}

\newtheorem{proposition}[theorem]{Proposition}

\newcommand{\Mod}[1]{\ (\mathrm{mod}\ #1)}

\newcommand{\si}{\supseteq_i}

\newtheorem{open}{Open Problem}

\begin{document}

\author[1]{Andrea Munaro\thanks{andrea.munaro@unipr.it}}
\author[2]{Shizhou Yang\thanks{syang22@qub.ac.uk}}
\affil[1]{Department of Mathematical, Physical and Computer Sciences, University of Parma, Italy}
\affil[2]{School of Mathematics and Physics, Queen’s University Belfast, United Kingdom}

\title{On algorithmic applications of sim-width and mim-width of $(H_1, H_2)$-free graphs}
\date{\today}

\maketitle

\begin{abstract} Mim-width and sim-width are among the most powerful graph width parameters, with sim-width more powerful than mim-width, which is in turn more powerful than clique-width. While several $\mathsf{NP}$-hard graph problems become tractable for graph classes whose mim-width is bounded and quickly computable, no algorithmic applications of boundedness of sim-width are known. In [Kang et al., A width parameter useful for chordal and co-comparability graphs, \textit{Theoretical Computer Science}, 704:1-17, 2017], it is asked whether \textsc{Independent Set} and \textsc{$3$-Colouring} are $\mathsf{NP}$-complete on graphs of sim-width at most $1$. We observe that, for each $k \in \mathbb{N}$, \textsc{List $k$-Colouring} is polynomial-time solvable for graph classes whose sim-width is bounded and quickly computable. Moreover, we show that if the same holds for \textsc{Independent Set}, then \textsc{Independent $\mathcal{H}$-Packing} is polynomial-time solvable for graph classes whose sim-width is bounded and quickly computable. This problem is a common generalisation of \textsc{Independent Set}, \textsc{Induced Matching}, \textsc{Dissociation Set} and \textsc{$k$-Separator}.   

We also make progress toward classifying the mim-width of $(H_1,H_2)$-free graphs in the case $H_1$ is complete or edgeless. Our results  solve some open problems in [Brettell et al., Bounding the mim-width of hereditary graph classes, \textit{Journal of Graph Theory}, 99(1):117-151, 2022]. 
\end{abstract}

\keywords{Width parameter, mim-width, sim-width, hereditary graph class, $\mathsf{XP}$ algorithm}

\section{Introduction}

Over the last decades, graph width parameters have proven to be an extremely successful tool in algorithmic graph theory. Arguably the most important reason explaining the jump from computational hardness of a graph problem to tractability, after restricting the input to some graph class~${\cal G}$, is that ${\cal G}$ has bounded ``width'', for some width parameter $p$. That is, there exists a constant $c$ such that, for each graph $G \in {\cal G}$, $p(G) \leq c$. A large number of width parameters have been introduced, and these parameters typically differ in strength. We say that a width parameter~$p$ {\it dominates} a width parameter~$q$ if there is a function~$f$ such that $p(G)\leq f(q(G))$ for all graphs~$G$. If~$p$ dominates~$q$ but $q$ does not dominate $p$, then~$p$ is said to be {\it more powerful} than~$q$. If both $p$ and $q$ dominate each other, then~$p$ and~$q$ are {\it equivalent}. For instance, the equivalent parameters boolean-width, clique-width, module-width, NLC-width and rank-width \cite{BTV11,Jo98,OS06,Ra08} are more powerful than the equivalent parameters branch-width, treewidth and mm-width \citep{CO00,JST18,RS91,Vat12} but less powerful than mim-width \cite{Vat12}, which is less powerful than sim-width \cite{KKST17}. We also mention that the recently introduced tree-independence number \citep{DMS21} is more powerful than treewidth, less powerful than sim-width and incomparable with both clique-width and mim-width (see \Cref{ourresults}). The \textit{tree-independence number of a graph $G$}, denoted $\tin(G)$, is defined as the minimum independence number over all tree decompositions of $G$, where the independence number of a tree decomposition of $G$ is the maximum independence number over all subgraphs of $G$ induced by some bag of the tree decomposition.

In this paper, we focus on mim-width and sim-width, both defined using the framework of branch decompositions. A \textit{branch decomposition} of a graph~$G$ is a pair $(T, \delta)$, where $T$ is a subcubic tree and $\delta$ is a bijection from~$V(G)$ to the leaves of $T$.  
Every edge $e \in E(T)$ partitions the leaves of $T$ into two classes, $L_e$ and $\overline{L_e}$, depending on which component of $T-e$ they belong to.
Hence, $e$ induces a partition $(A_e, \overline{A_e})$ of $V(G)$, where $\delta(A_e) = L_e$ and $\delta(\overline{A_e}) = \overline{L_e}$.
We let $G[A_e,\overline{A_e}]$ denote the bipartite subgraph of $G$ induced by the edges with one endpoint in $A_e$ and the other in $\overline{A_e}$.
A matching $F \subseteq E(G)$ of $G$ is {\it induced} if there is no edge in $G$ between vertices of different edges of $F$.
We let $\cutmim_{G}(A_{e}, \overline{A_{e}})$ denote the maximum size of an induced matching in $G[A_{e}, \overline{A_{e}}]$ and $\cutsim_{G}(A_{e}, \overline{A_{e}})$ denote the maximum size of an induced matching between $A_e$ and $\overline{A_{e}}$ in $G$ (equivalently, $\cutsim_{G}(A_{e}, \overline{A_{e}})$ is the maximum size of an induced matching in $G[A_{e}, \overline{A_{e}}]$ such that in addition there are no edges in $G$ between any two endpoints of matching edges that both belong to either $A_e$ or $\overline{A_e}$).   
The \emph{mim-width} of $(T, \delta)$, denoted $\mimw_{G}(T, \delta)$, is the maximum value of $\cutmim_{G}(A_{e}, \overline{A_{e}})$ over all edges $e\in E(T)$ and the \emph{mim-width} of $G$, denoted $\mimw(G)$, is the minimum value of $\mimw_{G}(T, \delta)$ over all branch decompositions $(T, \delta)$ of $G$. Similarly, the \emph{sim-width} of $(T, \delta)$, denoted $\simw_{G}(T, \delta)$, is the maximum value of $\cutsim_{G}(A_{e}, \overline{A_{e}})$ over all edges $e\in E(T)$ and the \emph{sim-width} of $G$, denoted $\simw(G)$, is the minimum value of $\simw_{G}(T, \delta)$ over all branch decompositions $(T, \delta)$ of $G$. Clearly, $\simw(G) \leq \mimw(G)$, for any graph $G$.  

We now briefly review the algorithmic implications of boundedness of mim-width, sim-width and tree-independence number. We begin with a recent and remarkable meta-theorem provided by \citet{BDJ22}. They showed that all problems expressible in $\mathsf{A\& C \ DN}$ logic, an extension of existential $\mathsf{MSO}_1$ logic, can be solved in $\mathsf{XP}$ time parameterized by the mim-width of a given branch decomposition of the input graph. This result, which can be viewed as the mim-width analogue of the famous meta-theorems for treewidth \citep{Cou90} and clique-width \citep{CMR00}, generalises essentially all the previously known $\mathsf{XP}$ algorithms parameterized by mim-width, as $\mathsf{A\& C \ DN}$ logic captures both local and non-local problems. Just to name few problems falling into this framework, we have all Locally Checkable Vertex Subset and Vertex Partitioning problems \citep{BV13,BTV13}, their distance versions \citep{JKST19} and their connectivity and acyclicity versions \citep{BK19}, \textsc{Longest Induced Path} and \textsc{Induced Disjoint Paths} \citep{JKT}, \textsc{Feedback Vertex Set} \citep{JKT20}, \textsc{Semitotal Dominating Set} \citep{GMR20}. Boundedness of tree-independence number has interesting algorithmic implications as well. \citet{DMS21} showed that, for any fixed finite set $\mathcal{H}$ of connected graphs, \textsc{Maximum Weight Independent $\mathcal{H}$-Packing}, a common generalisation of \textsc{Maximum Weight Independent Set} and \textsc{Maximum Weight Induced Matching} first defined in \citep{CH06}, can be solved in $\mathsf{XP}$ time parameterized by the independence number of a given tree decomposition of the input graph. They also showed that \textsc{$k$-Clique} and \textsc{List $k$-Colouring} admit linear-time algorithms for every graph class with bounded tree-independence number. This result holds more generally for every $(\tw,\omega)$-bounded graph class admitting a computable binding function, as shown by \citet{CZ17}, where a graph class $\mathcal{G}$ is \textit{$(\tw, \omega)$-bounded} if there exists a function $f$ (called a binding function) such that the treewidth of any graph $G \in \mathcal{G}$ is at most $f(\omega(G))$ and the same holds for all induced subgraphs of $G$. In \citep{DMS21}, it was observed that in every graph class with bounded tree-independence number, the treewidth is bounded by an explicit polynomial function of the clique number, and hence bounded tree-independence number implies $(\tw,\omega)$-boundedness.   

The trade-off of working with a more powerful width parameter is that, typically, fewer problems admit a polynomial-time algorithm when the parameter is bounded. Consider, for example, mim-width and the more powerful sim-width. \textsc{Dominating Set} is in $\mathsf{XP}$ parameterized by mim-width \citep{BTV13}. However, \textsc{Dominating Set} is $\mathsf{NP}$-complete on chordal graphs, a class of graphs of sim-width at most $1$ \citep{KKST17}. On the other hand, it is known that one can solve \textsc{Independent Set} and \textsc{$3$-Colouring} in polynomial time on both chordal graphs and co-comparability graphs, two classes of sim-width at most $1$, as shown by \citet{KKST17}. This led them to ask whether any of \textsc{Independent Set} and \textsc{$3$-Colouring} is $\mathsf{NP}$-complete on graphs of sim-width at most $1$ \citep[Question~2]{KKST17}. For convenience, we reformulate this question as follows: 

\begin{open}\label{opensim} Is any of \textsc{Independent Set} and \textsc{$3$-Colouring} in $\mathsf{XP}$ parameterized by the sim-width of a given branch decomposition of the input graph?
\end{open}

\noindent To the best of our knowledge, no problem $\mathsf{NP}$-complete on general graphs is known to be in $\mathsf{XP}$ parameterized by the sim-width of a given branch decomposition of the input graph.

In view of the discussion above, if we are interested in the computational complexity of a certain graph problem restricted to a special graph class, it is useful to know whether the mim-width of the class is bounded or not and, in the case of a positive answer to \Cref{opensim}, the same is true for sim-width. A systematic study on the boundedness of mim-width for hereditary graph classes, comparable to similar studies on the boundedness of clique-width (see, e.g., \citep{DJP19}) and treewidth \citep{LR22}, was recently initiated in \citep{BHMPP22} (see also \citep{BHMPP1}). Recall that a graph class is {\it hereditary} if it is closed under vertex deletion. It is well known that hereditary graph classes are exactly those classes characterised by a (unique) set ${\cal F}$ of minimal forbidden induced subgraphs. If $|{\cal F}|=1$ or $|{\cal F}|=2$, we say that the hereditary graph class is {\it monogenic} or {\it bigenic}, respectively.
In \citep{BHMPP22}, boundedness or unboundedness of mim-width has been determined for all monogenic classes and a large number of bigenic classes. 

In general, computing the mim-width is $\mathsf{NP}$-hard, deciding if the mim-width is at most~$k$ is $\mathsf{W}[1]$-hard when parameterized by $k$, and there is no polynomial-time algorithm for approximating the mim-width of a graph to within a constant factor of the optimal unless $\mathsf{NP} = \mathsf{ZPP}$~\cite{SV16}. Moreover, it remains a challenging open problem to obtain, for fixed $k$, a polynomial-time algorithm for computing a branch decomposition with mim-width $f(k)$ of a graph with mim-width~$k$; a similar problem for sim-width is open as well (see, e.g., \citep{Jaf20}). Therefore, in contrast to algorithms for graph classes of bounded treewidth or rank-width \citep{Bod96,OS06}, algorithms for classes of bounded mim-width require a branch decomposition of constant mim-width as part of the input. Obtaining such branch decompositions in polynomial time has been shown possible for several special graph classes~${\cal G}$ (see, e.g., \citep{BV13,BHMPP22}). In this case, we say that the mim-width of~${\cal G}$ is {\it quickly computable}.

Mim-width has proven to be particularly effective in tackling colouring problems. For instance, \citet{Kw20} showed the following (see also \citep{BHMP22}):

\begin{theorem}[\citet{Kw20}]\label{kwon}
For every $k\geq 1$, {\sc List $k$-Colouring} is polynomial-time solvable for every graph class whose mim-width is bounded and quickly computable.
\end{theorem}

Notice however that {\sc Colouring} (and hence {\sc List Colouring}) is $\mathsf{NP}$-complete for circular-arc graphs~\cite{GJMP80}, a class of graphs of mim-width at most~$2$ and for which mim-width is quickly computable~\cite{BV13}. The complexity of {\sc $k$-Colouring} restricted to $H$-free graphs has not yet been settled and there are infinitely many open cases when $H$ is a {\it linear forest}, that is, a disjoint union of paths. An extensive body of work has been devoted to studying whether forbidding certain linear forests makes \textsc{$k$-Colouring} and its generalisation \textsc{List $k$-Colouring} easy. We refer to~\cite{GJPS17} for a survey and to~\cite{CSZ20,HLS21,KMMNPS20} for updated summaries and briefly highlight below the connections with mim-width.  

For $r\geq 1$ and $s\geq 1$, let $K_{r,s}$ denote the complete bipartite graph with partition classes of size $r$ and $s$. The {\it $1$-subdivision} of a graph $G$ is the graph obtained from $G$ by subdividing each edge exactly once. The $1$-subdivision of $K_{1,s}$ is denoted by $K_{1,s}^1$; in particular $K_{1,2}^1=P_5$. \citet{BHMP22} showed that a number of known polynomial-time results for {\sc $k$-Colouring} and {\sc List $k$-Colouring} on hereditary classes \citep{CSZ20,CGKP15,GPS14b,HKLSS10} can be obtained, and strengthened, by combining \Cref{kwon} with the following:  

\begin{theorem}[\citet{BHMP22}]\label{t-new}
For every $r\geq 1$, $s\geq 1$ and $t\geq 1$, the mim-width of the class of $(K_r,K_{1,s}^1,P_t)$-free graphs is bounded and quickly computable.
\end{theorem}

The trivial but useful observation is that each yes-instance of {\sc List $k$-Colouring} is $K_{k+1}$-free, and so we obtain that, for every $k\geq 1$, $s\geq 1$ and $t\geq 1$, {\sc List $k$-Colouring} is polynomial-time solvable for $(K_{1,s}^1,P_t)$-free graphs \citep{BHMP22}. Hence, in the context of colouring problems on hereditary classes, it makes sense to investigate the mim-width of subclasses of $K_r$-free graphs. A first step is to consider the mim-width of $(K_r, H)$-free graphs, for some graph $H$. For any $H$ such that the mim-width of $(K_r,H)$-free graphs is bounded and quickly computable, \textsc{List $k$-Colouring} is polynomial-time solvable for all $k < r$. More generally, for problems admitting polynomial-time algorithms when mim-width is bounded and quickly computable, we obtain $\mathsf{XP}$ algorithms parameterized by $\omega(G)$ when restricted to $H$-free graphs. For example, \citet{CKPRS21} showed that for $P_5$-free graphs, there exists an $n^{O(\omega(G))}$-time algorithm for \textsc{Max Partial $H$-Colouring} (a common generalisation of \textsc{Maximum Independent Set} and \textsc{Odd Cycle Transversal} which is polynomial-time solvable when mim-width is bounded and quickly computable). \Cref{t-new} allows to generalise this, although with a worse running time (see \citep{BHMP22,CKPRS21}). 

From a merely structural point of view, the study of the mim-width of $(K_r, H)$-free graphs falls into the systematic study of the mim-width of bigenic classes mentioned above. For each $r \geq 4$, \citet{BHMPP22} completely classified the mim-width of the class of $(K_r, H)$-free graphs, except for one infinite family, and asked the following:

\begin{open}[\citet{BHMPP22}]\label{o-3}For each $r \ge 4$, and for each $t \geq 0$ and $u \geq 1$ such that $t+u \geq 2$, determine the (un)boundedness of mim-width of $(K_r, tP_2 + uP_3)$-free graphs.
\end{open}

Consider now the class of $(rP_1,H)$-free graphs. If the mim-width of such a class is bounded and quickly computable, we obtain, for many problems, $\mathsf{XP}$ algorithms parameterized by $\alpha(G)$ for the class of $H$-free graphs. For $r \geq 5$, \citet{BHMPP22} completely classified the mim-width of the class of $(rP_1, H)$-free graphs, except for one infinite family, and asked the following:

\begin{open}[\citet{BHMPP22}]\label{o-4}For each $r \ge 4$, and for each $s,t \geq 2$, determine the (un)boundedness of mim-width of $(rP_1, \overline{K_{s,t} + P_1})$-free graphs.
\end{open}

\begin{figure}[h!]
\centering
\includegraphics[scale=0.9]{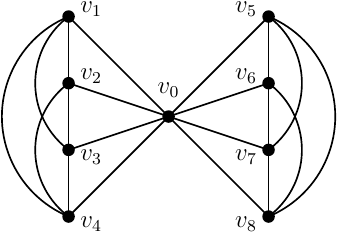}
\caption{The graph $\overline{K_{4,4}+P_1}$.}
\label{K44}
\end{figure}

\subsection{Our Results}\label{ourresults}

In this paper we observe that {\sc List $k$-Colouring} is polynomial-time solvable for every graph class whose sim-width is bounded and quickly computable, thus answering in the positive one half of \Cref{opensim}. We also show that if \textsc{Independent Set} is polynomial-time solvable for a given graph class whose sim-width is bounded and quickly computable, then the same is true for its generalisation \textsc{Independent $\mathcal{H}$-packing}. Finally, we completely resolve \Cref{o-4} and make considerable progress toward solving \Cref{o-3}. 

\subsubsection{Algorithmic implications of boundedness of sim-width}

Let us begin by discussing our results related to \Cref{opensim}. Let $K_t \boxminus K_t$ be the graph obtained from $2K_t$ by adding a perfect matching and let $K_t \boxminus S_t$ be the graph obtained from $K_t\boxminus K_t$ by removing all the edges in one of the complete graphs. Combining \Cref{kwon} with \citep[Proposition~4.2]{KKST17} stated below, we observe that \textsc{List $k$-Colouring} is in $\mathsf{XP}$ when parameterized by the sim-width of a given branch decomposition of the input graph. 

\begin{proposition}[see Proof of Proposition~4.2 in \citep{KKST17}]\label{cliquefree} Let $G$ be a graph with no induced subgraph isomorphic to $K_t \boxminus K_t$ and $K_t \boxminus S_t$ and let $(T,\delta)$ be a branch decomposition of $G$ with $\simw_{G}(T,\delta) = w$. Then $\mimw_{G}(T,\delta) \leq R(R(w + 1, t), R(t, t))$. 
\end{proposition}

\begin{theorem}\label{listsim}
For every $k\geq 1$, {\sc List $k$-Colouring} is polynomial-time solvable for every graph class whose sim-width is bounded and quickly computable.
\end{theorem}

\begin{proof} Given an instance consisting of a graph $G$ and a $k$-list assignment $L$, together with a branch decomposition $(T, \delta)$ of $G$ with $\simw_{G}(T,\delta) = w$, we proceed as follows. We check in polynomial time whether $G$ contains a copy of $K_{k+1}$. If it does, then we have a no-instance. Otherwise, $G$ is $K_{k+1}$-free. Then, by \Cref{cliquefree}, $(T, \delta)$ has mim-width at most $R(R(w + 1, k+1), R(k+1, k+1))$, and we simply apply \Cref{kwon}. This concludes the proof.    
\end{proof}

It is worth noticing that \Cref{listsim} does not really give wider applicability when compared to \Cref{kwon}. Indeed, input graphs of {\sc List $k$-Colouring} can always be assumed to be $K_{k+1}$-free and every subclass of $K_{k+1}$-free graphs has bounded sim-width if and only if it has bounded mim-width: This follows from \Cref{cliquefree} and the fact that $\simw(G) \leq \mimw(G)$ for any graph $G$. Nevertheless, \Cref{listsim} has interesting consequences. Besides answering in the positive one half of \Cref{opensim}, it extends the result in \citep{DMS21} that \textsc{List $k$-Colouring} is polynomial-time solvable for every graph class whose tree-independence number is bounded and quickly computable. This is because of the following unpublished observation of Dallard, Krnc, Kwon, Milani\v{c}, Munaro and \v{S}torgel, which is part of a work in progress and whose proof we sketch for convenience.

\begin{lemma}\label{simw-tin} Let $G$ be a graph. Then $\simw(G)\le \tin(G)$.
\end{lemma}

\begin{proof}[Proof sketch] Given a tree decomposition $(F, \{B_t\}_{t \in V(F)})$ of $G$, the proof of Proposition~3.1 in \citep{KKST17} shows how to construct a branch decomposition $(T, \delta)$ of $G$ such that, for each $e \in E(T)$, either $N_G(A_e) \cap \overline{A_e}$ or $N_G(\overline{A_e}) \cap A_e$ is contained in a bag in $\{B_t\}_{t \in V(F)}$. Consider then a tree decomposition $(F, \{B_t\}_{t \in V(F)})$ of $G$ with tree-independence number $\tin(G)$ and the corresponding branch decomposition $(T, \delta)$ of $G$ satisfying the property above. Fix $e \in E(T)$ and suppose without loss of generality that $N_G(A_e) \cap \overline{A_e} \subseteq B_t$, for some $t \in V(F)$. This implies that the independence number of $G[N_G(A_e) \cap \overline{A_e}]$ is at most $\tin(G)$ and so $\cutsim_G(A_e, \overline{A_e}) \leq \tin(G)$. Since this holds for every $e \in E(T)$, we have that $\simw_G(T, \delta) \le \tin(G)$ and so $\simw(G) \le \tin(G)$.
\end{proof}

Together with the fact that complete bipartite graphs have bounded sim-width (in fact, bounded clique-width) but unbounded tree-independence number \citep{DMS21}, \Cref{simw-tin} implies that sim-width is more powerful than tree-independence number. Note also that graph classes of bounded sim-width are not necessarily $(\tw,\omega)$-bounded and so \Cref{listsim} cannot be deduced from the results in \citep{CZ17}. Indeed, it is easy to see that complete bipartite graphs, which have bounded sim-width, are not $(\tw,\omega)$-bounded. However, we do not know whether a $(\tw,\omega)$-bounded graph class has necessarily bounded sim-width.  

In \Cref{algoimpli}, we show that a positive answer to the other half of \Cref{opensim} would have important algorithmic implications for \textsc{Maximum Weight Independent $\mathcal{H}$-Packing}, a problem studied for example in \citep{CH06,DMS21}. Before formulating it, we state some definitions and results. Let $\mathcal{H}$ be a set of connected graphs. Given a graph $G$, let $\mathcal{H}_G$ be the set of all subgraphs of $G$ isomorphic to a member of $\mathcal{H}$. The \textit{$\mathcal{H}$-graph of $G$}, denoted $\mathcal{H}(G)$, is defined in \citep{CH06} as follows: the vertex set is $\mathcal{H}_G$ and two distinct subgraphs of $G$ isomorphic to a
member of $\mathcal{H}$ are adjacent if and only if they either have a vertex in common or there is an edge
in $G$ connecting them. \citet{CH06} showed that, for any set $\mathcal{H}$ of connected graphs, the $\mathcal{H}$-graph of any chordal graph is chordal. \citet{DMS21} generalised this by showing that mapping any graph $G$ to
its $\mathcal{H}$-graph does not increase the tree-independence number. We show that this operation does not increase the sim-width either.

\begin{restatable}{theorem}{operation}\label{operation}
Let $\mathcal{H}$ be a non-empty finite set of connected non-null graphs and let $r$ be the maximum number of vertices of a graph in $\mathcal{H}$. Let $G$ be a graph and let $(T,\delta)$ be a branch decomposition of $G$. If $|V(\mathcal{H}(G))| > 1$, then we can obtain in $O(|V(G)|^{r+1})$ time a branch decomposition $(T',\delta')$ of $\mathcal{H}(G)$ such that $\simw_{\mathcal{H}(G)}(T',\delta') \leq \simw_G(T,\delta)$.
\end{restatable}

Two subgraphs $H_1$ and $H_2$ of a graph $G$ are \textit{independent} if they are vertex-disjoint and no edge of $G$ joins a vertex of $H_1$ with a vertex of $H_2$. An \textit{independent $\mathcal{H}$-packing} in $G$ is a set of pairwise independent subgraphs from $\mathcal{H}_G$. Given a graph $G$, a weight function $w\colon \mathcal{H}_G \rightarrow \mathbb{Q}_{+}$ on the subgraphs in $\mathcal{H}_G$, and an independent $\mathcal{H}$-packing $P$ in $G$, the weight of $P$ is defined as
$\sum_{H\in P}w(H)$. Given a graph $G$ and a weight function $w\colon \mathcal{H}_G \rightarrow \mathbb{Q}_{+}$, the \textsc{Maximum Weight Independent $\mathcal{H}$-Packing} problem asks to find an independent $\mathcal{H}$-packing in $G$ of maximum weight. If all subgraphs in $\mathcal{H}_G$ have weight $1$, we obtain the special case \textsc{Independent $\mathcal{H}$-Packing}. \textsc{Maximum Weight Independent $\mathcal{H}$-Packing} is a common generalisation of several problems studied in the literature, including \textsc{Maximum Weight Independent Set}, \textsc{Maximum Weight Induced Matching}, \textsc{Dissociation Set} and \textsc{$k$-Separator} (we refer to \citep{DMS21} for a comprehensive literature review). 

\citet{CH06} showed that \textsc{Independent $\mathcal{H}$-packing} is polynomial-time solvable, among others, for the following graph classes: weakly chordal graphs and hence chordal graphs, AT-free graphs and hence co-comparability graphs, circular-arc graphs, circle graphs. \citet{DMS21} showed that \textsc{Maximum Weight Independent $\mathcal{H}$-Packing} is polynomial-time solvable for every graph class whose tree-independence number is bounded and quickly computable. With the aid of \Cref{operation}, we show the following.  

\begin{restatable}{corollary}{reduction}\label{reduction} Let $\mathcal{H}$ be a non-empty finite set of connected non-null graphs such that each graph in $\mathcal{H}$ has at most $r$ vertices. Let $\mathcal{G}$ be a graph class whose sim-width is bounded and quickly computable. If \textsc{Maximum Weight Independent Set} is polynomial-time solvable for $\mathcal{G}$, then \textsc{Maximum Weight Independent $\mathcal{H}$-Packing} is polynomial-time solvable for $\mathcal{G}$. Similarly, if \textsc{Independent Set} is polynomial-time solvable for $\mathcal{G}$, then \textsc{Independent $\mathcal{H}$-Packing} is polynomial-time solvable for $\mathcal{G}$. 
\end{restatable}

\subsubsection{Mim-width of $(H_1,H_2)$-free graphs}

We now address the classification of (un)boundedness of mim-width of $(H_1,H_2)$-free graphs, where $H_1$ is either $rP_1$ or $K_r$. 

In \Cref{resolveedgeless}, we completely resolve \Cref{o-4} by showing the following.

\begin{restatable}{theorem}{edgeless}\label{edgeless}
Let $r \geq 3$ and $s, t \geq 2$ be integers. Then the mim-width of the class of $(rP_1, \overline{K_{s,t} + P_1})$-free graphs is bounded if and only if: 
\begin{itemize}
  \item $r = 3$ and one of $s$ and $t$ is at most $3$;
  \item $r = 4$ and one of $s$ and $t$ is at most $2$.
\end{itemize}
In all these cases, the mim-width is also quickly computable.
\end{restatable}

In \Cref{unboundedcomplete}, we finally address the case $H_1 = K_r$, related to \Cref{o-3}, by showing the following two results. 

\begin{restatable}{theorem}{completefive}\label{completefive} Let $r \geq 5$ be an integer and let $H = sP_1 + tP_2 + uP_3$, for $s,t,u \geq 0$. Then exactly one of the following holds:
\begin{itemize}
    \item $H \ssi sP_1+tP_2$, or $H \ssi sP_1+P_3$, and the mim-width of the class of $(K_r, H)$-free graphs is bounded and quickly computable;    
    \item $H \si P_3+P_2+P_1$, and the mim-width of the class of $(K_r, H)$-free graphs is unbounded;
    \item $H = 2P_3$, or $H = P_3+P_2$.
\end{itemize}
\end{restatable}

\begin{restatable}{theorem}{completefour}\label{completefour} Let $r = 4$ and let $H = sP_1 + tP_2 + uP_3$, for $s,t,u \geq 0$. Then exactly one of the following holds:
\begin{itemize}
  \item $H \ssi sP_1+tP_2$, or $H \ssi sP_1+P_3$, and the mim-width of the class of $(K_r, H)$-free graphs is bounded and quickly computable;    
    \item $H \si P_3+2P_2+P_1$, or $2P_3 + P_2$, and the mim-width of the class of $(K_r, H)$-free graphs is unbounded;   
    \item $H = P_3 + 2P_2$, or $H = P_3+P_2+sP_1$, or $H = 2P_3 + sP_1$.
  \end{itemize}
\end{restatable}

Our results are related to the class of $uP_3$-free graphs. Recently, \citet{HLS21} showed that, for every $u \geq 1$, \textsc{List $5$-Colouring} is polynomial-time solvable for $uP_3$-free graphs. Since an instance of \textsc{List $5$-Colouring} can always be assumed to be $K_6$-free, in view of \Cref{listsim} an alternative approach to obtaining the aforementioned result might pass through studying the sim-width of $(K_6, uP_3)$-free graphs. Unfortunately, \Cref{completefive} readily shows that, with the possible exception of the case $u=2$, this is not possible: For each $u \geq 3$, the mim-width of $(K_6, uP_3)$-free graphs is unbounded and, by \citep[Proposition~4.2]{KKST17}, the same must be true for sim-width.       

\section{Preliminaries}
\label{sec:basis}

We consider only finite graphs $G=(V,E)$ with no loops and no multiple edges. A graph is \textit{null} if it has no vertices. For a vertex $v \in V$, the \textit{neighbourhood} $N(v)$ is the set of vertices adjacent to $v$ 
in $G$.
The \textit{degree} $d(v)$ of a vertex $v \in V$ is the size $|N(v)|$ of its neighbourhood.
A graph is \textit{subcubic} if every vertex has degree at most~$3$.
For disjoint $S,T \subseteq V$, we say that $S$ is \textit{complete to} $T$ if every vertex of $S$ is adjacent to every vertex of $T$, and $S$ is \textit{anticomplete to} $T$ if there are no edges between $S$ and $T$. The \textit{distance} from a vertex $u$ to a vertex $v$ in $G$ is the length of a shortest path between $u$ and $v$. A set $S\subseteq V$ {\it induces} the subgraph 
$G[S]=(S,\{uv\; :\; u,v\in S, uv\in E\})$.
If $G'$ is an induced subgraph of $G$, we write $G'\ssi G$. 
The \textit{complement} of $G$ is the graph $\overline{G}$ with vertex set $V(G)$, such that $uv \in E(\overline{G})$ if and only if $uv \notin E(G)$.

The \textit{$k$-subdivision} of an edge $uv$ in a graph replaces $uv$ by $k$ new vertices $w_{1}, \dots, w_{k}$ with edges $uw_{1}, w_{k}v$ and $w_{i}w_{i+1}$ for each $i \in \{1, \dots, k-1\}$, i.e. the edge is replaced by a path of length $k + 1$. 
The \textit{disjoint union} $G+H$ of graphs $G$ and $H$ has vertex set $V(G) \cup V(H)$ and edge set $E(G) \cup E(H)$.
We denote the disjoint union of $k$ copies of $G$ by $kG$. 
For a graph $H$, a graph $G$ is {\it $H$-free} if $G$ has no induced subgraph isomorphic to $H$. 
For a set of graphs $\{H_1,\ldots,H_k\}$, a graph $G$ is {\it $(H_1,\ldots,H_k)$-free} if $G$ is $H_i$-free for every $i\in \{1,\ldots,k\}$. 

Let $T$ be a tree and let $v$ be a leaf of $T$. Let $u$ be a vertex of degree at least $3$ having shortest distance in $T$ from $v$ and let $P$ be the $v,u$-path in $T$. The operation of \textit{trimming} the leaf $v$ consists in deleting from $T$ the vertex set $V(P) \setminus \{u\}$.   

An \textit{independent set} of a graph $G$ is a set of pairwise non-adjacent vertices and the maximum size of an independent set of $G$ is denoted by $\alpha(G)$. 
A \textit{clique} of a graph $G$ is a set of pairwise adjacent vertices and the maximum size of a clique of $G$ is denoted by $\omega(G)$. 
A \textit{matching} of a graph is a set of edges with no shared endpoints.

The path and the complete graph on $n$ vertices are denoted by $P_n$ and $K_{n}$, respectively. A graph is \textit{$r$-partite}, for $r \geq 2$, if its vertex set admits a partition into $r$ classes such that every edge has its endpoints in different classes. An $r$-partite graph in which every two vertices from different partition classes are adjacent is a \textit{complete $r$-partite graph} and a $2$-partite graph is also called \textit{bipartite}. The complete bipartite graph with partition classes of size $t$ and $s$ is denoted by $K_{t,s}$. A graph is \textit{co-bipartite} if it is the complement of a bipartite graph.

For $\ell \geq 1$, an {\em $\ell$-caterpillar} is a subcubic tree $T$ on $2\ell$ vertices with $V(T) = \{ s_1,\ldots,s_{\ell},t_1,\ldots,t_{\ell} \}$,
such that $E(T) = \{ s_it_i \;:\; 1 \leq i \leq \ell \} \cup \{ s_is_{i+1} \;:\; 1 \leq i \leq \ell -1 \}$. The vertices $t_1,t_2,\dotsc,t_{\ell}$ are the leaves and the path $s_1s_2\cdots s_{\ell}$ is the \textit{backbone} of the caterpillar.

A {\em colouring} of a graph $G=(V,E)$ is a mapping $c\colon V\rightarrow\{1,2,\ldots \}$ that gives each vertex~$u\in V$ a {\it colour} $c(u)$ in such a way that, for every two adjacent vertices $u$ and $v$, we have that $c(u)\neq c(v)$. If for every $u\in V$ we have $c(u)\in \{1,\ldots,k\}$, then we say that $c$ is a {\it $k$-colouring} of $G$.
The {\sc Colouring} problem is to decide whether a given graph $G$ has a $k$-colouring for some given integer $k\geq 1$. If $k$ is {\it fixed}, that is, not part of the input, we call this the $k$-{\sc Colouring} problem. It is well known that $k$-{\sc Colouring} is $\mathsf{NP}$-complete for each $k \geq 3$. A generalisation of {\sc $k$-Colouring} is the following.
For an integer $k\geq 1$, a  {\it $k$-list assignment} of a graph
$G=(V,E)$ is a function $L$ that assigns each vertex $u\in V$ a {\it list} $L(u)\subseteq \{1,2,\ldots,k\}$ of {\it admissible} colours for $u$. A colouring $c$ of $G$ {\it respects} $L$ if  $c(u)\in L(u)$ for every $u\in V$. For a fixed integer~$k\geq 1$, the {\sc List $k$-Colouring} problem is to decide whether a given graph~$G$ with a $k$-list assignment $L$ admits a colouring that respects $L$. 
By setting $L(u)=\{1,\ldots,k\}$ for every $u\in V$, we obtain the {\sc $k$-Colouring} problem. 

\section{Sim-width and independent packings}\label{algoimpli}

In this section we show \Cref{operation} and \Cref{reduction}. Let $\mathcal{H}$ be a finite set of connected non-null graphs. Given a graph $G$, let $\mathcal{H}_G$ be the set of all subgraphs of $G$ isomorphic to a member of $\mathcal{H}$. Recall that the $\mathcal{H}$-graph of $G$, denoted $\mathcal{H}(G)$, is defined as follows: the vertex set is $\mathcal{H}_G$ and two distinct subgraphs of $G$ isomorphic to a
member of $\mathcal{H}$ are adjacent if and only if they either have a vertex in common or there is an edge
in $G$ connecting them. We begin by showing \Cref{operation}: mapping a graph $G$ to its $\mathcal{H}$-graph does not increase the sim-width.

\operation*

\begin{proof} Observe that if $G$ is edgeless, then $\mathcal{H}(G)$ is edgeless as well and the statement trivially holds. Therefore, we assume that $G$ is not edgeless, and hence $\simw_{G}(T,\delta) \geq 1$.

Let $\mathcal{H}=\{H_1,\ldots,H_n\}$. Let $h$ be an arbitrary vertex of $\mathcal{H}(G)$. Hence, $h$ corresponds to a subgraph of $G$ isomorphic to $H_i$, for some $i \in \{1,\ldots,n\}$. This means there exists a unique vertex set $S(h) \subseteq V(G)$ such that $|S(h)| = |V(H_i)|$ and $G[S(h)]$ contains a copy of $H_i$ as a subgraph ($S(h)$ is just the vertex set of the subgraph of $G$ corresponding to $h$). We compute all $S(h)$, for $h \in \mathcal{H}(G)$, in $O(|V(G)|^{r})$ time as follows.  We enumerate all $O(|V(G)|^{r})$ subsets of vertices of $G$ of size at most $r$. For each such set $S$ and each $H \in \mathcal{H}$ with $|S|$ vertices, we iterate over all $|S|! \leq r! = O(1)$ possible bijections $g\colon V(H) \rightarrow S$. We then keep the subsets $S$ for which one such bijection maps every pair of adjacent vertices in $H$ to a pair of adjacent vertices in
$G[S]$. We now arbitrarily order $V(G)$ and let $f(h)$ be the smallest vertex in $S(h)$ with respect to this ordering. For $v \in V(G)$, let $F(v) = \{h \in V(\mathcal{H}(G)) : f(h) = v\}$. Note that $F(v)$ is a clique in $\mathcal{H}(G)$. We can compute all sets $F(v)$, for $v \in V(G)$, in $O(|V(G)|\cdot |V(\mathcal{H}(G))|) = O(|V(G)|^{r+1})$ time.   

We are now ready to construct $(T',\delta')$ from $(T,\delta)$ as follows (see \Cref{fig:Hpack}). For each leaf $t \in V(T)$, we let $v_t = \delta^{-1}(t)$, and do the following. If $F(v_t) \neq \varnothing$, we distinguish two cases. Suppose first that $|F(v_t)| = 1$. In this case, build a $|F(v_t)|$-caterpillar $C_t$ and add the edge connecting the single vertex $x_t$ in the backbone of $C_t$ with the node $t$. Suppose now that $|F(v_t)| \geq 2$. In this case, build a $|F(v_t)|$-caterpillar $C_t$, subdivide an arbitrary edge of the backbone of $C_t$ by adding a new vertex $x_t$ and add the edge $x_tt$. Finally, if $F(v_t) = \varnothing$, trim the leaf $t$ of $T$, as defined in \Cref{sec:basis}. Observe that, since $|V(\mathcal{H}(G))| > 1$, either there exists a leaf $t \in V(T)$ such that $|F(v_t)| \geq 2$ or there exist at least two leaves $t_1,t_2 \in V(T)$ such that $|F(v_{t_1})| \geq 1$ and $|F(v_{t_2})| \geq 1$. This implies that each leaf $t$ of $T$ such that $F(v_t) = \varnothing$ can be trimmed. Moreover, by definition, no new leaf is created after an application of trimming. Let $T'$ be the tree obtained by the procedure above. Let $\delta'$ be the map from $V(\mathcal{H}(G))$ to the leaves of $T'$ which restricted to $F(v_t)$ is an arbitrary bijection from $F(v_t)$ to the leaves of $C_t$. It is easy to see that $(T',\delta')$ is a branch decomposition of $\mathcal{H}(G)$ and that it can be computed in $O(|V(G)|^2)$ time.

We now show that $\simw_{\mathcal{H}(G)}(T',\delta') \leq \simw_G(T,\delta)$. Suppose that $\simw_{\mathcal{H}(G)}(T',\delta') = k$. Since the statement is trivially true if $k \leq 1$, we may assume $k \geq 2$. Each $e' \in E(T')$ naturally induces a partition $(A_{e'}, \overline{A_{e'}})$ of $V(\mathcal{H}(G))$. Consider $e \in E(T')$ such that $\cutsim_{\mathcal{H}(G)}(A_{e}, \overline{A_{e}}) = \simw_{\mathcal{H}(G)}(T', \delta') = k$. Then, there is a matching $\{x_1'y_1',\ldots,x_k'y_k'\}$ of size $k$ such that $\{x_1',\ldots,x_k'\} \subseteq A_{e}$ and $\{y_1',\ldots,y_k'\}\subseteq \overline{A_{e}}$ are independent sets of $\mathcal{H}(G)$. Suppose first that $e$ is an edge of $C_t$ or the edge $x_tt$, for some leaf $t\in V(T)$. Then, one of $A_e$ and $\overline{A_{e}}$ is a subset of $F(v_t)$, where $v_t = \delta^{-1}(t)$. Since each $F(v_t)$ is a clique in $\mathcal{H}(G)$, we have that $k \leq 1$. Hence, we may assume that $e \in E(T') \cap E(T)$. Then, for any $h \in V(\mathcal{H}(G))$, $\delta'(h)$ and $\delta(f(h))$ belong to the same component of $T'-e$, and so $e$ naturally induces a partition $(A_e, \overline{A_e})$ of $V(\mathcal{H}(G))$ and a partition $(B_e,\overline{B_e})$ of $V(G)$ satisfying the following property: For any $h \in V(\mathcal{H}(G))$, $h \in A_e$ if and only if $f(h) \in B_e$.
 
We claim that, for $i \neq j$, $S(x_i') \cup S(y_i')$ and $S(x_j') \cup S(y_j')$ are disjoint and anticomplete in $G$. Indeed, suppose without loss of generality that $S(x_i')$ shares a vertex with either $S(x_j')$ or $S(y_j')$. Then, $x_i'$ is adjacent to either $x_j'$ or $y_j'$ in $\mathcal{H}(G)$, a contradiction. Similarly, if there is an edge between $S(x_i')$ and either $S(x_j')$ or $S(y_j')$ in $G$, then $x_i'$ is adjacent to either $x_j'$ or $y_j'$ in $\mathcal{H}(G)$, a contradiction again.
 
We now claim that $G[S(x_i') \cup S(y_i')]$ is connected. Since $G[S(x_i')]$ contains a copy of a connected graph $H_s \in \mathcal{H}$, with $|S(x_i')| = |V(H_s)|$, as a subgraph, we have that $G[S(x_i')]$ is connected. Similarly, $G[S(y_i')]$ is connected. Moreover, since $x_i'$ is adjacent to $y_i'$, either $S(x_i')$ shares a vertex with $S(y_i')$ or there is an edge in $G$ between $S(x_i')$ and $S(y_i')$. In either case we obtain that $G[S(x_i') \cup S(y_i')]$ is connected. 

Therefore, for each $i \in \{1,\ldots,k\}$, there is a path $P_i$ in $G[S(x_i') \cup S(y_i')]$ from $f(x_i')$ to $f(y_i')$ in $G$, say $P_i =  v_0v_1\cdots v_\ell$ where $v_0 = f(x_i')$ and $v_{\ell} = f(y_i')$. Since  $x_i' \in A_e$ and $y_i' \in \overline{A_e}$, it follows that $f(x_i') \in B_e$ and $f(y_i') \in \overline{B_e}$. Since the path $P_i$ must cross the cut $(B_e, \overline{B_e})$ of $G$, there exists $q \in \{0,\ldots,\ell-1\}$ such that $v_q \in B_e$ and $v_{q+1} \in \overline{B_e}$. 
We let $x_i = v_q$ and $y_i = v_{q+1}$. Clearly, $x_iy_i \in E(G)$. We now claim that, for each $i\neq j$, $\{x_i,y_i\}$ and $\{x_j,y_j\}$ are disjoint and anticomplete in $G$. This simply follows from the fact that, for $p \in \{i,j\}$, $\{x_p,y_p\} \subseteq G[S(x_p') \cup S(y_p')]$ and $G[S(x_i') \cup S(y_i')]$ and $G[S(x_j') \cup S(y_j')]$ are disjoint and anticomplete in $G$.

Let now $X = \{x_1,\ldots,x_k\}$ and $Y=\{y_1,\ldots,y_k\}$. By the previous paragraph, $X \subseteq B_e$ and $Y \subseteq \overline{B_e}$, $X$ and $Y$ are independent sets and $G[X,Y] \cong kP_2$. Therefore, $\simw_{G}(T,\delta) \geq\cutsim_G(B_e,\overline{B_e}) \geq k = \simw_{\mathcal{H}(G)}(T',\delta')$. 
\end{proof}

\begin{figure}[h!]
\centering
\includegraphics[scale=0.9]{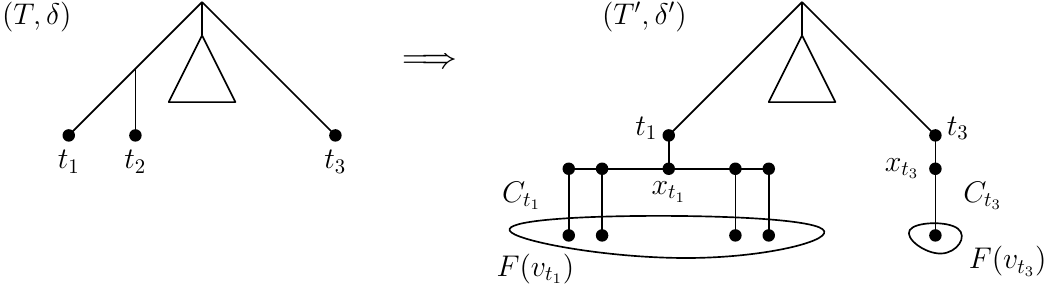}
\caption{How to construct a branch decomposition $(T', \delta')$ of $\mathcal{H}(G)$ from a branch decomposition $(T, \delta)$ of $G$. We distinguish vertices $t_i$ such that $|F(v_{t_i})| = 0$ ($i = 2$), $|F(v_{t_i})| = 1$ ($i = 3$) and $|F(v_{t_i})| \geq 2$ ($i = 1$).}
\label{fig:Hpack}
\end{figure}

Recall that two subgraphs $H_1$ and $H_2$ of a graph $G$ are independent if they are vertex-disjoint and no edge of $G$ joins a vertex of $H_1$ with a vertex of $H_2$. An independent $\mathcal{H}$-packing in $G$ is a set
of pairwise independent subgraphs from $\mathcal{H}_G$. Given a graph $G$, a weight function $w\colon \mathcal{H}_G \rightarrow \mathbb{Q}_{+}$ on the subgraphs in $\mathcal{H}_G$, and an independent $\mathcal{H}$-packing $P$ in $G$, the weight of $P$ is defined as
the sum $\sum_{H\in P}w(H)$. Given a graph $G$ and a weight function $w\colon \mathcal{H}_G \rightarrow \mathbb{Q}_{+}$, \textsc{Maximum Weight Independent $\mathcal{H}$-Packing} is the problem of finding an independent $\mathcal{H}$-packing in $G$ of maximum weight. Besides \Cref{operation}, in order to show \Cref{reduction}, we need the following two results. 

\begin{theorem}[\citep{DMS21}]\label{repre} Let $\mathcal{H}$ be a non-empty finite set of connected non-null graphs and let $r$ be the maximum number of vertices of a graph in $\mathcal{H}$. Then there exists an algorithm that takes as input a graph $G$ and computes the graph $\mathcal{H}(G)$ in $O(|V(G)|^{2r})$ time.
\end{theorem}

\begin{observation}[\citep{DMS21}]\label{red} Let $\mathcal{H}$ be a finite set of connected non-null graphs. Let $G$ be a graph and let $w\colon \mathcal{H}_G \rightarrow \mathbb{Q}_{+}$. Let $I$ be an independent set in $\mathcal{H}(G)$
of maximum weight with respect to the weight function $w$. Then $I$ is an independent $\mathcal{H}$-packing in
$G$ of maximum weight.
\end{observation}

\reduction*

\begin{proof} Given the input graph $G \in \mathcal{G}$, we compute in polynomial time a branch decomposition of $G$ of sim-width at most $k$, for some integer $k$. We then compute $\mathcal{H}(G)$ in polynomial time using \Cref{repre}. If $|V(\mathcal{H}(G))| \leq 1$, we immediately conclude thanks to \Cref{red}. Otherwise, by \Cref{operation}, we compute in polynomial time a branch decomposition of $\mathcal{H}(G)$ of sim-width at most $k$. Finally, using the assumed algorithm, we compute in polynomial time a maximum-weight independent set in $\mathcal{H}(G)$ which, by \Cref{red}, is an independent $\mathcal{H}$-packing in $G$ of maximum weight.   
\end{proof}

\section{Mim-width of $(rP_1,\overline{K_{t,s}+P_1})$-free graphs}\label{resolveedgeless}

In this section we show the mim-width dichotomy for the class of $(rP_1,\overline{K_{t,s}+P_1})$-free graphs stated in \Cref{edgeless}. We begin by identifying the cases of bounded mim-width (\Cref{boundededgeless}) and then pass to the cases of unbounded mim-width (\Cref{unboundededgeless}). These results are then combined to prove \Cref{edgeless} (\Cref{dichotomy}).  

\subsection{Boundedness results}\label{boundededgeless}

In this section we show that, for each $t\geq 4$, the mim-width of $(3P_1,\overline{K_{3,t}+P_1})$-free graphs and the mim-width of $(4P_1,\overline{K_{2,t}+P_1})$-free graphs are bounded and quickly computable (\Cref{33tt1,42tt1}, respectively). The proofs are based on the following common strategy. We find $t$ pairwise non-adjacent vertices $v_1, \ldots v_t$ in the input graph $G$ ($t = 2$ in \Cref{33tt1} and $t = 3$ in \Cref{42tt1}). We then obtain a partition of $V(G)$ where one partition class is $\{v_1, \ldots, v_t\}$ and the remaining ones are the sets of private neighbours of subsets of $\{v_1, \ldots, v_t\}$ with respect to $\{v_1, \ldots, v_t\}$. We finally construct an appropriate branch decomposition of $G$ and use the following simple observation.    

\begin{observation}\label{boundmim}
Let $V_1,\ldots,V_m$ be a partition of $V(G)$ and let $(T,\delta)$ be a branch decomposition of $G$. Then, 
\begin{equation*}
\mimw_G(T,\delta) = \max_{e \in E(T)} \cutmim_G(A_e,\overline{A_e}) \leq  \max_{e \in E(T)}\sum_{1 \leq i,j \leq m} \cutmim_G(A_e \cap V_i,\overline{A_e} \cap V_j).
\end{equation*}
\end{observation}

We will need two auxiliary results. The first one below is left as an easy exercise (see \Cref{fig:newbranch}).

\begin{lemma}\label{addleaves} Let $G$ be a graph and let $(T, \delta)$ be a branch decomposition of $G$ with $\mimw_{G}(T, \delta) \leq k$, with $k\geq 1$. Let $G'$ be the graph obtained from $G$ by adding a vertex of degree at most $1$. Then we can construct in $O(1)$ time a branch decomposition $(T', \delta')$ of $G'$ with $\mimw_{G'}(T', \delta') \leq k$.  
\end{lemma}

\begin{figure}[h!]
\centering
\includegraphics[scale=0.8]{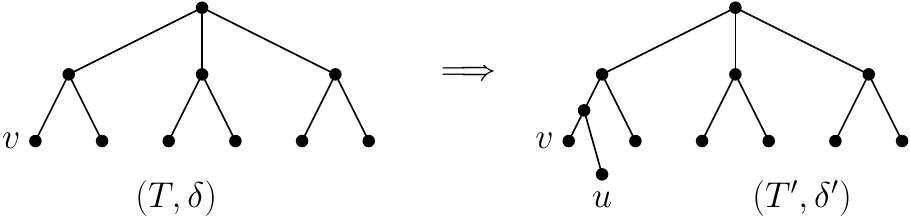}
\caption{How to construct a branch decomposition $(T', \delta')$ of $G'$ from a branch decomposition $(T, \delta)$ of $G$, where $G'$ is obtained from $G$ by adding a leaf vertex $u$ adjacent to $v$.}
\label{fig:newbranch}
\end{figure}

The second one is essentially stated in the proof of \citep[Corollary~3.7.4]{Vat12}. We provide its short proof for completeness.

\begin{lemma}[\citet{Vat12}]\label{paths} Let $G$ be a graph with $|V(G)| > 1$ and maximum degree at most $2$. Then $\mimw(G) \leq 2$ and a branch decomposition $(T, \delta)$ of $G$ with $\mimw_{G}(T, \delta) \leq 2$ can be constructed in $O(n)$ time. 
\end{lemma}

\begin{proof}
Suppose that $G$ has $k$ components, $C_1,\dots,C_k$, where each $C_i$ is a path or a cycle with vertex set $\{v_{i,1},\ldots,v_{i,|C_i|}\}$. For $1 < j < |C_i|$, each $v_{i,j}$ is adjacent to $v_{i,j-1}$ and $v_{i,j+1}$ and, if $C_i$ is a cycle, $v_{i,1}$ is adjacent to $v_{1,|C_i|}$. For each component $C_i$, we construct a $|C_i|$-caterpillar $T_i$ with leaves $\ell_{i,1},\ldots,\ell_{i,|C_i|}$ and subdivide an arbitrary edge of the backbone of $T_i$ with a new vertex $t_i$, unless the backbone of $T_i$ has size $1$, in which case we let $t_i$ be the unique vertex of the backbone. We then construct a $k$-caterpillar $T_0$ with leaves $\ell_{0,1},\ldots,\ell_{0,k}$. Let $T$ be the subcubic tree obtained from the disjoint union of $T_0,T_1,\ldots,T_k$ by adding the edges $\ell_{0,1}t_1,\ldots,\ell_{0,k}t_k$ and, if $k=1$, by additionally deleting $V(T_0)$. Let $\delta$ be the bijection from the vertices of $G$ to the leaves of $T$ given by $\delta(v_{i,j})= \ell_{i,j}$. Clearly, $(T,\delta)$ is a branch decomposition of $G$ and it can be constructed in $O(n)$ time.

We now show that $\mimw_{G}(T, \delta) \leq 2$. Let $e \in E(T)$ and consider the partition $(A_e,\overline{A_e})$ of $V(G)$ induced by $e$. Suppose first that $e$ belongs to $E(T_0)$ or $e = \ell_{0,j}t_j$ for some $j$. Then, for each component $C_i$ of $G$, $V(C_i)$ is fully contained in either $A_e$ or $\overline{A_e}$ and so $\cutmim_{G}(A_{e}, \overline{A_{e}}) = 0$. Suppose now that $e$ belongs to the backbone of $T_i$, for some $i > 0$. Then, it is easy to see that there are at most two edges across the cut $(A_e,\overline{A_e})$, from which $\cutmim_G(A_e,\overline{A_e}) \leq 2$. Suppose finally that $e$ is incident to a leaf of $T$. Then $\cutmim_G(A_e,\overline{A_e}) = 1$. These observations imply that $\mimw_{G}(T, \delta) \leq 2$.
\end{proof}

We can finally provide our two boundedness results. In both proofs, we make repeated implicit use of Ramsey's theorem: there exists a least positive integer $R(r, s)$ for which every graph with at least $R(r, s)$ vertices either contains an independent set of size $r$ or a clique of size $s$. Observe that, for $r,s > 1$, $R(r, s) \geq s$.

\begin{theorem}
\label{33tt1} Let $t\geq 4$ and let $G$ be a $(3P_1,\overline{K_{3,t}+P_1})$-free graph. Then $\mimw(G) < 5R(3,t)+8t+46$ and a branch decomposition $(T, \delta)$ of $G$ with $\mimw_{G}(T, \delta) < 5R(3,t)+8t+46$ can be constructed in $O(n^2)$ time.
\end{theorem}

\begin{proof} We assume that $G$ contains two non-adjacent vertices $v_a$ and $v_b$, or else $G$ is a complete graph and the statement is trivially true. Let $S_z = \{v_a,v_b\}$. Since $G$ is $3P_1$-free, all remaining vertices are adjacent to at least one of $v_a$ and $v_b$ and we partition them into three classes $S_a,S_b$ and $S_{ab}$ as follows: $S_a$ is the set of vertices that are adjacent to $v_a$ but not $v_b$, $S_b$ is the set of vertices that are adjacent to $v_b$ but not $v_a$ and $S_{ab}$ is the set of vertices that are adjacent to both $v_a$ and $v_b$. Note that $S_a$ is a clique, or else two non-adjacent vertices in $S_a$ together with $v_b$ would induce a copy of $3P_1$. Similarly, $S_b$ is a clique.

We now proceed to the construction of a branch decomposition of $G$ by distinguishing two cases. We say that $G$ is \textit{good} (w.r.t. $\{v_a, v_b\}$) if every vertex in $S_a$ has at most two neighbours in $S_b$ and every vertex in $S_b$ has at most two neighbours in $S_a$. Otherwise, we say that $G$ is \textit{bad} (w.r.t. $\{v_a, v_b\}$).

Suppose first that $G$ is good. Then, $G[S_a,S_b]$ has maximum degree at most $2$ and, if $G[S_a,S_b]$ contains at least two vertices, \Cref{paths} allows us to construct a branch decomposition $(T_1,\delta_1)$ of $G[S_a,S_b]$ with mim-width at most $2$. Let $u$ be a leaf of $T_1$ and let $e$ be the edge of $T_1$ incident to $u$. We subdivide $e$ by introducing a new vertex $x$ and obtain a new tree $T_1'$. If however $G[S_a,S_b]$ contains exactly one vertex, let $x$ be this vertex. We now let $\ell = |V(G) \setminus (S_a \cup S_b)|$ and consider an $\ell$-caterpillar $T_2$ (notice that $\ell \geq 2$). We subdivide one of the edges of the backbone of $T_2$ by introducing a new vertex $y$ and obtain a new tree $T_2'$. Let $\delta_2$ be any bijection from $V(G) \setminus (S_a \cup S_b)$ to the set of leaves of $T'_2$. We finally add the edge $xy$ in order to obtain a subcubic tree $T$, unless $G[S_a,S_b]$ is the null graph, in which case we let $T = T_2'$. Clearly, the set of leaves $L$ of $T$ is the disjoint union of the set of leaves of $T_1$ and the set of leaves of $T_2$. Considering the map $\delta\colon V(G) \rightarrow L$ which coincides with $\delta_1$ when restricted to $S_a \cup S_b$ and with $\delta_2$ when restricted to $V(G) \setminus (S_a \cup S_b)$, we obtain a branch decomposition $(T, \delta)$ of $G$. If $G$ is bad, we simply let $(T, \delta)$ be any branch decomposition of $G$. 

The branch decomposition $(T, \delta)$ of $G$ defined above can be constructed in $O(n^2)$ time. Indeed, we first find two non-adjacent vertices $v_a$ and $v_b$ in $O(n^2)$ time and check whether $G[S_a, S_b]$ has maximum degree at most $2$ in linear time. If so, $G$ is good and we then construct $(T,\delta)$ in $O(n)$ time thanks to \Cref{paths}. Otherwise, $G$ is bad, and we trivially construct $(T,\delta)$ in linear time.

\begin{claim}
\label{33tl1}
Let $S_P$ and $S_Q$ be subsets of vertices of $G$, not necessarily disjoint. If there exists a vertex that is complete to both $S_P$ and $S_Q$, then $\cutmim_G(A_e \cap S_P, \overline{A_e} \cap S_Q) < R(3,t) + 6$, for any $e \in E(T)$.
\end{claim}

\begin{claimproof}[Proof of \Cref{33tl1}] Let $v \in V(G)$ be complete to $S_P$ and $S_Q$. Suppose, to the contrary, that $\cutmim_G(A_e \cap S_P, \overline{A_e} \cap S_Q) \geq R(3,t) + 6$ for some $e \in E(T)$ and let $\{p_1q_1, \ldots, p_{R(3,t)+6}q_{R(3,t)+6}\}$ be an induced matching witnessing this, where $\{p_1,\ldots,p_{R(3,t)+6}\} \subseteq A_e \cap S_P$ and $\{q_1,\ldots,q_{R(3,t)+6}\} \subseteq \overline{A_e} \cap S_Q$. Since $G$ is $3P_1$-free, $\{q_1,\ldots,q_{R(3,t)}\}$ contains a clique of size at least $t$. Without loss of generality, $\{q_1,\ldots,q_{t}\}$ induces such a clique. Observe now that $\{p_{R(3,t)+1},\ldots,p_{R(3,t)+6}\}$ contains a clique of size $3$, as $R(3,3)=6$. Without loss of generality, $\{p_{R(3,t)+1},p_{R(3,t)+2},p_{R(3,t)+3}\}$ induces such a clique. But then we have that $G[p_{R(3,t)+1},p_{R(3,t)+2},p_{R(3,t)+3},q_1,q_2,\ldots,q_t,v] \cong \overline{K_{3,t}+P_1}$, a contradiction.
\end{claimproof}

\begin{claim}
\label{33tl2} Suppose that $G$ is bad. Then $\cutmim_G(A_e \cap S_a, \overline{A_e} \cap S_b) < 4t$ and $\cutmim_G(A_e \cap S_b, \overline{A_e} \cap S_a) < 4t$, for any $e \in E(T)$.
\end{claim}

\begin{claimproof}[Proof of \Cref{33tl2}] By symmetry, it is enough to show the first statement. Since $G$ is bad, $G[S_a, S_b]$ contains a vertex $u$ of degree at least $3$. Without loss of generality, $u \in S_a$. Suppose, to the contrary, that $\cutmim_G(A_e \cap S_a, \overline{A_e} \cap S_b) \geq 4t$ for some $e \in E(T)$ and let $\{a_1b_1, \ldots, a_{4t}b_{4t}\}$ be an induced matching witnessing this, where $ \{a_1,\ldots,a_{4t}\} \subseteq A_e \cap S_a$ and $\{b_1,\ldots,b_{4t}\} \subseteq \overline{A_e} \cap S_{b}$. Let $v_1,v_2,v_3 \in S_b$ be distinct neighbours of $u \in S_a$. Observe now that all except possibly $t-1$ vertices in $\{a_1,\ldots,a_{4t}\}$ are adjacent to at least one of $v_1,v_2,v_3$, or else there are $t$ vertices in $\{a_1,\ldots,a_{4t}\}$, say without loss of generality $a_1,\ldots,a_t$, non-adjacent to any of $v_1,v_2,v_3$ and so, since $S_a$ and $S_b$ are cliques, $G[v_1,v_2,v_3,a_1,\ldots,a_t,u] \cong \overline{K_{3,t}+P_1}$, a contradiction. Hence, there is a vertex in $\{v_1,v_2,v_3\}$ with at least $t$ neighbours in $\{a_1,\ldots,a_{4t}\}$, say without loss of generality $v_1$ is adjacent to $a_1,\ldots,a_{t}$, and so $G[b_{t+1},b_{t+2},b_{t+3},a_1,\ldots,a_t,v_1] \cong \overline{K_{3,t}+P_1}$, a contradiction.
\end{claimproof}

We can finally show that $\mimw_G(T,\delta) < 5R(3,t)+8t+46$. Let $D = \{a,b,ab,z\}$. Since $S_a, S_b, S_{ab}, S_z$ is a partition of $V(G)$, \Cref{boundmim} implies that $$\mimw_G(T,\delta) \leq \max_{e \in E(T)} \sum_{i,j \in D}\cutmim_G(A_e \cap S_i,\overline{A_e} \cap S_j).$$ It is then enough to estimate the terms in the sum. Since $S_a$ and $S_b$ are cliques, $\cutmim_G(A_e \cap S_a, \overline{A_e} \cap S_a) \leq 1$ and $\cutmim_G(A_e \cap S_b, \overline{A_e} \cap S_b) \leq 1$. Moreover, since $v_a$ is complete to $S_a$ and $S_{ab}$, and $v_b$ is complete to $S_b$ and $S_{ab}$, \Cref{33tl1} implies that $\cutmim_G(A_e \cap S_a, \overline{A_e} \cap S_{ab})$, $\cutmim_G(A_e \cap S_b, \overline{A_e} \cap S_{ab})$, $\cutmim_G(A_e \cap S_{ab}, \overline{A_e} \cap S_{a})$, $\cutmim_G(A_e \cap S_{ab}, \overline{A_e} \cap S_{b}), \cutmim_G(A_e \cap S_{ab}, \overline{A_e} \cap S_{ab}) < R(3,t)+6$. Observe now that, for any $i \in D$, $\cutmim_G(A_e \cap S_z, \overline{A_e} \cap S_{i}) \leq 2$, $\cutmim_G(A_e \cap S_i, \overline{A_e} \cap S_{z}) \leq 2$. 

It remains to bound $\cutmim_G(A_e \cap S_a, \overline{A_e} \cap S_b)$ and $\cutmim_G(A_e \cap S_b, \overline{A_e} \cap S_{a})$. If $G$ is bad then, by \Cref{33tl2}, $\cutmim_G(A_e \cap S_a, \overline{A_e} \cap S_b) < 4t$ and $\cutmim_G(A_e \cap S_b, \overline{A_e} \cap S_a) < 4t$. If $G$ is good, we proceed as follows. Suppose first that either $e = xy$ or $e \in E(T_2')$. Then all vertices of $S_a$ and $S_b$ belong to the same partition class of $V(G)$ induced by $e$ and so $\cutmim_G(A_e \cap S_a, \overline{A_e} \cap S_b) = \cutmim_G(A_e \cap S_b, \overline{A_e} \cap S_a) = 0$. Suppose finally that $e \in E(T_1')$. Then $e$ induces a partition $(A'_e,\overline{A'_e})$ of $S_a \cup S_b$ with respect to $(T_1,\delta_1)$, and $(A'_e,\overline{A'_e})$ coincides with $(A_e,\overline{A_e})$ restricted to $S_a \cup S_b$. Consequently, $\cutmim_G(A_e \cap S_a, \overline{A_e} \cap S_b) = \cutmim_G(A'_e \cap S_a, \overline{A'_e} \cap S_b) \leq 2$ as $\cutmim_G(T_1,\delta_1) \leq 2$. The same holds for $\cutmim_G(A_e \cap S_b, \overline{A_e} \cap S_a)$.

By the previous paragraphs, $\mimw_{G}(T, \delta) < 2\cdot 1 + 5\cdot (R(3,t)+6) + 7\cdot 2 + 2\cdot 4t  = 5R(3,t)+8t+46$.
\end{proof}

\begin{theorem}
\label{42tt1}
Let $t\geq 4$ and let $G$ be a $(4P_1,\overline{K_{2,t}+P_1})$-free graph. Then $\mimw(G) < 43R(4,t)+24t+208$ and a branch decomposition $(T, \delta)$ of $G$ with $\mimw_{G}(T, \delta) < 43R(4,t)+24t+214$ can be computed in $O(n^3)$ time.
\end{theorem}

\begin{proof} We assume that $G$ contains three pairwise non-adjacent vertices $v_a$, $v_b$ and $v_c$, or else $G$ is $3P_1$-free and the statement follows from \Cref{33tt1}. Since $G$ is $4P_1$-free, all remaining vertices are adjacent to at least one of $v_a$, $v_b$ and $v_c$. For a subset $\alpha \subseteq \{a,b,c\}$, let $S_\alpha = \cap_{i\in \alpha}N(v_i) \setminus \cup_{j\in \{a,b,c\} \setminus \alpha}N(v_j)$. In words, $S_{\alpha}$ is the set of private neighbours of $\{v_i : i \in \alpha\}$ with respect to $\{v_a,v_b,v_c\}$. Note that $S_a$, $S_b$ and $S_c$ are cliques, or else, for distinct $i,j,k \in \{a, b, c\}$, two non-adjacent vertices in $S_i$ together with $v_j$ and $v_k$ would induce a copy of $4P_1$. This fact will be repeatedly used in the claims below. For $\alpha,\beta \subseteq \{a,b,c\}$ and an integer $s \geq 1$, we say that the vertex set $S_\alpha$ is \textit{$3s$-almost-complete} to the vertex set $S_\beta$ if there are at most two vertices in $S_\alpha$ non-adjacent to at least $3s$ vertices in $S_\beta$. 

\begin{claim}
\label{42tl2}
Let $p, q \in \{a,b,c\}$ with $p \neq q$. If a vertex in $S_p$ is adjacent to at least two vertices in $S_q$, then $S_q$ is $3t$-almost-complete to $S_p$. 
\end{claim}

\begin{claimproof}[Proof of \Cref{42tl2}] Note that $v_p$ is complete to $S_p$ but anticomplete to $S_q$ and $v_q$ is complete to $S_q$ but anticomplete to $S_p$. Suppose that $x \in S_p$ is adjacent to two distinct vertices $y_1$ and $y_2$ of $S_q$. Then $\{y_1, y_2\} \cap \{v_q\} = \varnothing$. 

We claim that there are at most $t-1$ vertices in $S_p$ anticomplete to $\{y_1,y_2\}$. Indeed, if there are $t$ vertices in $S_p$ anticomplete to $\{y_1,y_2\}$, then these $t$ vertices together with $\{x,y_1,y_2\}$ induce a copy of $\overline{K_{2,t}+P_1}$, as $S_p$ and $S_q$ are cliques, a contradiction. 

Let now $y \in S_q$ be a vertex distinct from $y_1$ and $y_2$. We claim that $y$ is anticomplete to at most $t-1$ vertices in $S_p \cap N(y_i)$, for each $i \in \{1,2\}$. Indeed, if there are $t$ vertices in $S_p\cap N(y_i)$ anticomplete to $y$, then these $t$ vertices together with $\{y_i,v_q,y\}$ induce a copy of $\overline{K_{2,t}+P_1}$, a contradiction. 

Let $A_1 = S_p \cap N(y_1)$, $A_2 = S_p \cap N(y_2)$ and let $y \in S_q$ be a vertex distinct from $y_1$ and $y_2$. Clearly, $S_p = A_1 \cup A_2 \cup (S_p \setminus (A_1 \cup A_2))$. By the second paragraph, $|S_p \setminus (A_1 \cup A_2)| \leq t-1$ and so $y$ is anticomplete to at most $t-1$ vertices in $S_p \setminus (A_1 \cup A_2)$. By the third paragraph, $y$ is anticomplete to at most $t-1$ vertices in $A_1$ and at most $t-1$ vertices in $A_2$. Therefore, $y$ is anticomplete to at most $3(t-1) < 3t$ vertices in $S_p$ and so $S_q$ is $3t$-almost-complete to $S_p$. 
\end{claimproof}

We now proceed to the construction of a branch decomposition of $G$. Consider first the graph $G_1$ with vertex set $V(G_1) =  S_a \cup S_b \cup S_c$ and edge set $E(G_1) = \{uv:uv \in E(G),u \in S_\alpha,v\in S_\beta, \alpha,\beta \in \{a,b,c\},\alpha \neq \beta, S_\alpha$ is not $3t$-almost-complete to $S_\beta, S_\beta$ is not $3t$-almost complete to $S_\alpha\}$. We claim that each vertex $v$ of $G_1$ has degree at most $2$. By symmetry, suppose that $v \in S_a$. By definition of $G_1$, $v$ has no neighbours in $S_a$. If $S_b$ is $3t$-almost-complete to $S_a$, then $v$ has no neighbours in $S_b$. Otherwise, $S_b$ is not $3t$-almost-complete to $S_a$ and, by \Cref{42tl2}, $v$ has at most one neighbour in $S_b$. Similarly, $v$ has at most one neighbour in $S_c$. Therefore, $G_1$ has maximum degree at most $2$ and so, by \Cref{paths}, if $G_1$ contains at least two vertices, then we can construct in $O(n)$ time a branch decomposition $(T_1, \delta_1)$ of $G_1$ with $\mimw_{G_1}(T_1, \delta_1) \leq 2$. 

For $x \in \{a,b,c\}$ and $Y = \{a,b,c\} \setminus \{x\}$, a vertex $v \in S_Y$ is $S_x$-\textit{good} if it has at most one neighbour in $S_x$, and $S_x$-\textit{bad} otherwise. Let $S^{\star}_Y$ be the set of vertices in $S_Y$ that are $S_x$-bad.
We now build a graph $G_2$ as follows. Start with $G_2 = G_1$. For each $x \in \{a,b,c\}$, let $Y = \{a,b,c\} \setminus \{x\}$. For each vertex $v \in S_Y$, if $v$ is $S_x$-good, then add $v$ to $V(G_2)$ and, if $v$ has a neighbour $u$ in $S_x$, add $uv$ to $E(G_2)$. In other words, we grow $G_1$ by adding leaf vertices or isolated vertices. 

Now, if $G_2$ is the null graph, let $T_2'$ be the null tree, and if $G_2$ consists of one vertex, let $T_2'$ be the tree with a single vertex $r$. Otherwise, $G_2$ contains at least two vertices and, given $(T_1, \delta_1)$,   we can construct a branch decomposition $(T_2, \delta_2)$ of $G_2$ with $\mimw_{G_2}(T_2, \delta_2) \leq 2$ in $O(n)$ time thanks to \Cref{addleaves}, unless $G_1$ contains at most one vertex, in which case $G_2$ has maximum degree at most $1$ and we let $(T_2, \delta_2)$ be any branch decomposition of $G_2$. We then subdivide one of the edges of $T_2$ by introducing a new vertex $r$ to obtain a new tree $T'_2$. Clearly, $\mimw_{G_2}(T'_2, \delta_2) = \mimw_{G_2}(T_2, \delta_2) \leq 2$. Let now $\ell = |V(G) \setminus V(G_2)|$ and consider an $\ell$-caterpillar $T_3$ (notice that $\ell \geq 3$). Let $\delta_3$ be any bijection from $V(G) \setminus V(G_2)$ to the set of leaves of $T_3$. We subdivide one of the edges of the backbone of $T_3$ by introducing a new vertex $s$ and obtain a new tree $T'_3$. We finally add the edge $rs$ in order to obtain a tree $T$. Observe that the set of leaves $L$ of $T$ is the disjoint union of the set of leaves $L_2$ of $T'_2$ and the set of leaves $L_3$ of $T'_3$. Considering the map $\delta$ which coincides with $\delta_i$ when restricted to $L_i$ (for $i = 2,3$), we obtain a branch decomposition $(T, \delta)$ of $G$. 

We now analyse the running time to construct $(T,\delta)$. Finding three pairwise non-adjacent vertices $v_a$, $v_b$ and $v_c$ and computing $S_\alpha$ for each $\alpha \subseteq \{a,b,c\}$ can be done in $O(n^3)$ time. Checking for $3t$-almost-completeness and constructing $G_1$ can be done in $O(n)$ time. Finding the $S_x$-good vertices and constructing $G_2$ can be done in $O(n)$ time. Therefore, constructing $(T,\delta)$ can be done in $O(n^3)$ time.

\begin{claim}\label{obs} Let $\alpha,\beta \subseteq \{a,b,c\}$. If $S_\alpha$ is $3t$-almost-complete to $S_\beta$, then $\cutmim_G(A_e \cap S_\alpha, \overline{A_e} \cap S_\beta) < 3t+1$ and $\cutmim_G(A_e \cap S_\beta, \overline{A_e} \cap S_\alpha) < 3t+1$, for any $e \in E(T)$.
\end{claim}

\begin{claimproof}[Proof of \Cref{obs}] Suppose that there exist $V_\alpha \subseteq A_e \cap S_\alpha$ and $V_\beta \subseteq \overline{A_e} \cap S_\beta$ such that $G[V_\alpha, V_\beta] \cong (3t+1)P_2$. Then, each of the $3t+1$ vertices in $V_\alpha$ is non-adjacent to at least $3t$ vertices in $V_{\beta}$, contradicting the fact that $S_\alpha$ is $3t$-almost-complete to $S_\beta$. The proof of the other inequality is similar.
\end{claimproof}

\begin{claim}
\label{42tl4}
Let $x \in \{a,b,c\}$ and $Y = \{a,b,c\} \setminus \{x\}$. Then $\cutmim_G(A_e \cap S_x, \overline{A_e} \cap S^{\star}_Y) < R(4,t)+t+1$ and $\cutmim_G(A_e \cap S^{\star}_Y, \overline{A_e} \cap S_x) < R(4,t)+t+1$, for any $e \in E(T)$.
\end{claim}

\begin{claimproof}[Proof of \Cref{42tl4}] We show the first inequality, the proof of the other being similar. Suppose, to the contrary, that there exists $e\in E(T)$ such that $\cutmim_G(A_e \cap S_x, \overline{A_e} \cap S^{\star}_Y) \geq R(4,t)+t+1$. Let $\{p_1q_1, \ldots, p_{R(4,t)+t+1}q_{R(4,t)+t+1}\}$ be an induced matching witnessing this, where $P = \{p_1,\ldots,p_{R(4,t)+t+1}\} \subseteq S_x$ and $Q = \{q_1,\ldots,q_{R(4,t)+t+1}\} \subseteq S^{\star}_Y$. Since $q_1$ is $S_x$-bad, let $u_1 \in S_x$ be one of its neighbours distinct from $p_1$. Suppose that $q_1$ has at least $R(4,t)$ neighbours in $Q$. Then, at least $t$ of these neighbours induce a clique. Without loss of generality, suppose that $\{q_2,\ldots,q_{t+1}\}$ are neighbours of $q_1$ inducing a clique. If $\{q_2,\ldots,q_{t+1}\}$ is anticomplete to $u_1$, then these $t$ vertices together with $\{q_1,p_1,u_1\}$ induce a copy of $\overline{K_{2,t}+P_1}$, a contradiction. Hence, $u_1$ has at least one neighbour in $\{q_2,\ldots,q_{t+1}\}$, say without loss of generality $q_2$. But then, $G[q_1,q_2,p_3,\ldots,p_{t+2},u_1] \cong \overline{K_{2,t}+P_1}$, a contradiction. 

Hence, $q_1$ has less than $R(4,t)$ neighbours in $Q$. Without loss of generality, suppose that $q_{R(4,t)+1},\ldots,q_{R(4,t)+t+1}$ are non-neighbours of $q_1$. Then, these $t+1$ vertices form a clique, or else two non-adjacent vertices $v$ and $v'$ among them would give $G[v,v',q_1,v_x] \cong 4P_1$, a contradiction. Next, since $q_{R(4,t)+1}$ is $S_x$-bad, it has another neighbour $u_2 \in S_x$ distinct from $p_{R(4,t)+1}$. Suppose that $\{q_{R(4,t)+2},\ldots,q_{R(4,t)+t+1}\}$ is anticomplete to $u_2$. Then, we have that $G[p_{R(4,t)+1},u_2,q_{R(4,t)+2},\ldots,q_{R(4,t)+t+1},q_{R(4,t)+1}] \cong \overline{K_{2,t}+P_1}$, a contradiction. Therefore, $u_2$ has at least one neighbour in $\{q_{R(4,t)+2},\ldots,q_{R(4,t)+t+1}\}$, say without loss of generality $q_{R(4,t)+2}$. Then, $G[q_{R(4,t)+1},q_{R(4,t)+2},p_1,\ldots,p_t,u_2] \cong \overline{K_{2,t}+P_1}$, a contradiction.
\end{claimproof}

\begin{claim}
\label{42tl1}
Let $\alpha, \beta \subseteq \{a,b,c\}$ with $\alpha \cap \beta \neq \varnothing$. Then $\cutmim_G(A_e \cap S_\alpha, \overline{A_e} \cap S_\beta) < R(4,t)+4$, for any $e \in E(T)$.
\end{claim}

\begin{claimproof}[Proof of \Cref{42tl1}] Let $i\in \alpha \cap \beta$. Then $v_i$ is complete to $S_\alpha$ and $S_\beta$. Suppose, to the contrary, that there exists $e \in E(T)$ such that $\cutmim_G(A_e \cap S_\alpha, \overline{A_e} \cap S_\beta) \geq R(4,t)+4$. Let $\{p_1q_1, \ldots, p_{R(4,t)+4}q_{R(4,t)+4}\}$ be an induced matching witnessing this, where $P = \{p_1,\ldots,p_{R(4,t)+4}\} \subseteq S_\alpha$ and $Q = \{q_1,\ldots,q_{R(4,t)+4}\} \subseteq S_\beta$. Since $G$ is $4P_1$-free, $Q$ contains a clique of size at least $t$. Without loss of generality, suppose that $\{q_1,\ldots,q_{t}\}$ induces a clique. Observe now that $\{p_{R(4,t)+1},\ldots,p_{R(4,t)+4}\}$ contains a pair of adjacent vertices, as $G$ is $4P_1$-free. Without loss of generality, suppose that $p_{R(4,t)+1}$ is adjacent to $p_{R(4,t)+2}$. But then, $G[p_{R(4,t)+1},p_{R(4,t)+2},q_1,q_2,\ldots,q_t,v_i] \cong \overline{K_{2,t}+P_1}$, a contradiction.
\end{claimproof}

We can finally show that $\mimw_G(T,\delta) < 43R(4,t)+24t+214$. Let $S_z = \{v_a,v_b,v_c\}$ and let $D = \{\{a\},\{b\},\{c\},\{a,b\},\{b,c\},\{a,c\},\{a,b,c\},\{z\}\}$. Since $\{S_{\alpha} : \alpha \in D\}$ is a partition of $V(G)$, \Cref{boundmim} implies that 
\begin{equation}\label{ubmimw}
\mimw_G(T,\delta) \leq \max_{e \in E(T)} \sum_{\alpha,\beta \in D}\cutmim_G(A_e \cap S_\alpha,\overline{A_e} \cap S_\beta).
\end{equation}

For any $\alpha \in D$, we have that $\cutmim_G(A_e \cap S_z,\overline{A_e} \cap S_\alpha) \leq 3$ and $\cutmim_G(A_e \cap S_\alpha,\overline{A_e} \cap S_z) \leq 3$, for any $e \in E(T)$. For $\alpha,\beta \neq \{z\}$, there are 49 distinct pairs $(\alpha,\beta)$. 12 of such pairs are such that $\alpha \cap \beta = \varnothing$: $(\{a\}.\{b\}),(\{b\},\{c\}),(\{a\},\{c\}),(\{a\},\{b,c\}),(\{b\},\{a,c\}),(\{c\},\{a,b\})$ and those obtained by swapping $\alpha$ and $\beta$. The remaining 37 pairs are such that $\alpha \cap \beta \neq \varnothing$. In this case, by \Cref{42tl1}, $\cutmim_G(A_e \cap S_\alpha, \overline{A_e} \cap S_\beta) \leq R(4,t)+4$, for any $e \in E(T)$.
 
We now estimate the terms in the sum above corresponding to pairs $(\alpha, \beta)$ such that $\alpha \cap \beta = \varnothing$. Suppose first that $(\alpha,\beta)$ is one of $(\{a\},\{b\}),(\{b\},\{c\}),(\{a\},\{c\}), (\{b\},\{a\}),(\{c\},\{b\}),(\{c\},\{a\})$. If $S_\alpha$ is $3t$-almost-complete to $S_\beta$ or $S_\beta$ is $3t$-almost-complete to $S_\alpha$ then, by \Cref{obs}, $\cutmim_G(A_e \cap S_\alpha, \overline{A_e} \cap S_\beta) < 3t+1$ and $\cutmim_G(A_e \cap S_\beta, \overline{A_e} \cap S_\alpha) < 3t+1$. Otherwise, $S_\alpha$ is not $3t$-almost-complete to $S_\beta$ and $S_\beta$ is not $3t$-almost-complete to $S_\alpha$. By definition of $G_1$ and $G_2$, this implies that $G[S_\alpha,S_\beta] = G_1[S_\alpha,S_\beta] = G_2[S_\alpha,S_\beta]$. If either $e = rs$ or $e$ belongs to $T_3'$, then all vertices of $S_\alpha$ and $S_\beta$ belong to the same partition class of $V(G)$ induced by $e$ and so $\cutmim_G(A_e \cap S_\alpha, \overline{A_e} \cap S_\beta) = \cutmim_G(A_e \cap S_\beta, \overline{A_e} \cap S_\alpha) = 0$. Otherwise, $e$ must belong to $T_2'$. The edge $e$ then induces a partition $(A'_e,\overline{A'_e})$ of the vertices of $G_2$ with respect to $(T'_2,\delta_2)$, and $(A'_e,\overline{A'_e})$ coincides with $(A_e,\overline{A_e})$ restricted to $S_\alpha \cup S_\beta$. Hence, $\cutmim_G(A_e \cap S_\alpha, \overline{A_e} \cap S_\beta) = \cutmim_{G_2}(A'_e \cap S_\alpha, \overline{A'_e} \cap S_\beta) \leq 2$.
 
Suppose finally that $(\alpha, \beta)$ is one of $(\{a\},\{b,c\})$, $(\{b\},\{a,c\})$, $(\{c\},\{a,b\})$, $(\{b,c\},\{a\})$, \linebreak $(\{a,c\},\{b\})$, $(\{a,b\},\{c\})$. Clearly, $\cutmim_G(A_e \cap S_\alpha, \overline{A_e} \cap S_\beta) \leq \cutmim_G(A_e \cap S_\alpha, \overline{A_e} \cap S^{\star}_\beta)+ \cutmim_G(A_e \cap S_\alpha, \overline{A_e} \cap (S_\beta \setminus S^{\star}_\beta))$. Note that $G[S_\alpha,S_\beta \setminus S^{\star}_\beta] = G_2[S_\alpha,S_\beta \setminus S^{\star}_\beta]$. Thus, by the same reasoning as in the previous paragraph, $\cutmim_G(A_e \cap S_\alpha, \overline{A_e} \cap (S_\beta \setminus S^{\star}_\beta)) \leq 2$. On the other hand, by \Cref{42tl4}, $\cutmim_G(A_e \cap S_\alpha, \overline{A_e} \cap  S^{\star}_\beta) \leq R(4,t)+t+1$. Therefore, $\cutmim_G(A_e \cap S_\alpha, \overline{A_e} \cap S_\beta) \leq R(4,t)+t+3$. 
 
Combining these bounds with \labelcref{ubmimw}, we obtain $\mimw_G(T, \delta) < 14\cdot 3 + 37\cdot (R(4,t)+4) + 6\cdot (3t+1) +6\cdot (R(4,t)+t+3)= 43R(4,t)+24t+214$.
\end{proof}

\subsection{Unboundedness results}\label{unboundededgeless}

All the unboundedness results of this section are obtained by applying the same strategy. The class of walls plays a crucial role. A \emph{wall of height $h$ and width $r$} (an \emph{$(h \times r)$-wall} for short) is the graph obtained from the grid of height~$h$ and width~$2r$ as follows. Let $C_1, \dots, C_{2r}$ be the set of vertices in each of the $2r$ columns of the grid, in their natural left-to-right order. For each column $C_{j}$, let $e_{1}^{j}, e_{2}^{j}, \dots, e_{h-1}^{j}$ be the edges between two vertices of $C_j$, in their natural top-to-bottom order. If $j$ is odd, we delete all edges $e_{i}^{j}$ with $i$ even. If $j$ is even, we delete all edges $e_{i}^{j}$ with $i$ odd. We then remove all vertices of the resulting graph whose degree is~$1$. This final graph is an \textit{elementary $(h \times r)$-wall} (see \Cref{fig:walls}). We denote by $\mathcal{W}$ the class of all elementary $2n \times 2n$ walls, for $n \geq 1$.

\begin{theorem}[\citet{BHMPP22}]\label{walls} Let $W$ be an elementary $n\times n$ wall with $n \geq 7$. Then $\mimw(W) \geq \frac{\sqrt{n}}{50}$. Hence, $\mathcal{W}$ has unbounded mim-width.
\end{theorem}

The idea is to start from an elementary wall, find an appropriate vertex colouring, and repeatedly apply the following result (the case $k =2$ was first proved in \citep{Men17}).

\begin{lemma}[\citet{BHMPP22}]\label{bipcompl} Let $G$ be a $k$-partite graph with partition classes $V_1,\ldots,V_k$ and let $G'$ be a
graph obtained from $G$ by adding edges where, for each added edge, there exists some $i$ such that both endpoints are in $V_i$. Then $\mimw(G') \geq \frac{1}{k}\cdot\mimw(G)$.
\end{lemma}

\begin{figure}[h!]
\begin{subfigure}[b]{0.3\textwidth}
\captionsetup{justification=raggedright}
\includegraphics[scale=0.6]{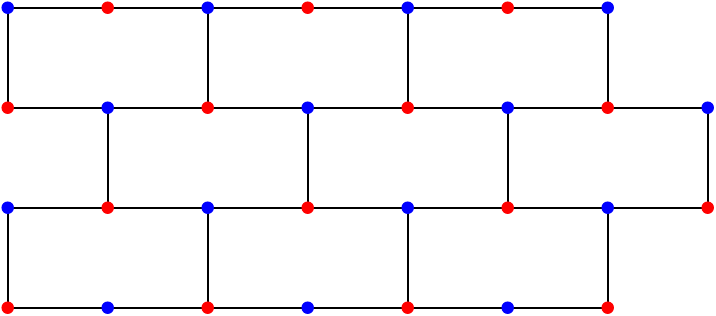}
\caption{A $2$-colouring of the elementary $4\times 4$ wall.}
\label{fig:2col}
\end{subfigure}
\hspace{3.9cm}
\begin{subfigure}[b]{0.3\textwidth}
\captionsetup{justification=raggedright}
\includegraphics[scale=0.6]{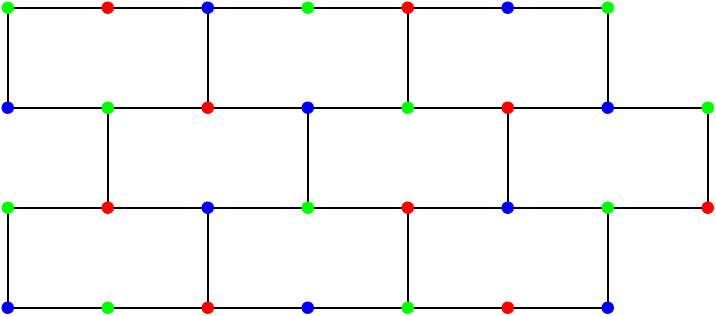}
\caption{A $3$-colouring of the elementary $4\times 4$ wall.}
\label{fig:3col}
\end{subfigure}

\vspace{0.6cm}

\centering
\begin{subfigure}[b]{0.3\textwidth}
\captionsetup{justification=raggedright}
\includegraphics[scale=0.6]{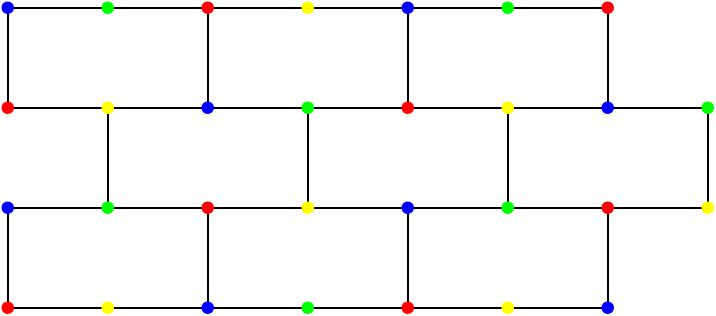}
\caption{A $4$-colouring of the elementary $4\times 4$ wall.}
\label{fig:4col}
\end{subfigure}
\caption{The different colourings of elementary walls used in the proofs of \Cref{344t1,433t1,522t1}.} 
\label{fig:walls}
\end{figure} 

\begin{theorem}
\label{344t1}
The class of $(3P_1,\overline{K_{4,4}+P_1})$-free graphs has unbounded mim-width.
\end{theorem}

\begin{proof}
Let $W$ be an elementary $2n\times 2n$ wall and consider its proper $2$-colouring depicted in \Cref{fig:2col}. We add edges within each colour class to make them cliques. Let $f(W)$ be the graph obtained and let $\mathcal{W}_1 = \{f(W):W \in \mathcal{W}\}$. By \Cref{walls} and \Cref{bipcompl}, $\mathcal{W}_1$ has unbounded mim-width. 

Note that, for the graph $f(W)$, every two vertices of the same colour are adjacent, and every two vertices of different colours are adjacent if and only if they are adjacent in $W$. Clearly, $f(W)$ is $3P_1$-free. It remains to show that $f(W)$ is $\overline{K_{4,4}+P_1}$-free.

\begin{claim}
\label{344c1}
Any copy of $K_5$ in $f(W)$ is monochromatic. 
\end{claim}

\begin{claimproof}[Proof of \Cref{344c1}] Let $u_1, \ldots, u_5$ be the vertices of a copy of $K_5$. Since $f(W)$ is obtained from $W$ by adding edges within each colour class, if an edge $uv \in E(f(W))$ is not monochromatic, then $uv$ belongs to $E(W)$ as well. Hence, there cannot be one blue vertex and four red vertices in $\{u_1,\ldots,u_5\}$, since this would imply that in $W$ there is a vertex with four neighbours. Also, there cannot be exactly two blue vertices in $\{u_1,\ldots,u_5\}$, for otherwise these two blue vertices share three common red neighbours in $W$, contradicting the fact that in $W$ any two vertices have at most one common neighbour. By symmetry, there cannot be exactly one or two red vertices, and so $\{u_1,\ldots,u_5\}$ is monochromatic.
\end{claimproof}

Suppose, to the contrary, that $f(W)$ contains an induced copy of $\overline{K_{4,4}+P_1}$ with vertex set $\{v_0, \ldots, v_8\}$ as depicted in \Cref{K44}. By \Cref{344c1}, the two copies of $K_5$ induced by $\{v_0,v_1,v_2,v_3,v_4\}$ and $\{v_0,v_5,v_6,v_7,v_8\}$ must both be monochromatic. Hence, $v_1,\ldots,v_8$ must be of the same colour. This implies that $v_1,\ldots,v_8$ must form a clique in $f(W)$, a contradiction.
\end{proof}

\begin{theorem}
\label{433t1}
The class of $(4P_1,\overline{K_{3,3}+P_1})$-free graphs has unbounded mim-width.
\end{theorem}

\begin{proof}
Let $W$ be an elementary $2n\times 2n$ wall and consider its proper $3$-colouring depicted in \Cref{fig:3col}. We add edges within each colour class to make them cliques. Let $g(W)$ be the graph obtained and let $\mathcal{W}_2 = \{g(W) : W \in \mathcal{W}\}$. By \Cref{walls} and \Cref{bipcompl}, $\mathcal{W}_2$ has unbounded mim-width. Clearly, $g(W)$ is $4P_1$-free. It remains to show that $g(W)$ is $\overline{K_{3,3}+P_1}$-free.

\begin{claim}
\label{433c1}
Any copy of $K_4$ in $g(W)$ is monochromatic. 
\end{claim}

\begin{claimproof}[Proof of \Cref{433c1}] Let $u_1, \ldots, u_4$ be the vertices of a copy of $K_4$. At least two such vertices have the same colour, say colour $c$. Since $g(W)$ is obtained from $W$ by adding edges within each colour class, if an edge $uv \in E(g(W))$ is not monochromatic, then $uv$ belongs to $E(W)$ as well. Observe first that there cannot be exactly two vertices with colour $c$ in $\{u_1,\ldots,u_4\}$, for otherwise these two vertices coloured $c$ have two common neighbours coloured different from $c$, contradicting the fact that in $W$ any two vertices have at most one common neighbour. Moreover, there cannot be exactly three vertices coloured $c$ in $\{u_1,\ldots,u_4\}$, since this would imply that in $W$ there is a vertex not coloured $c$ adjacent to three vertices coloured $c$. However, in the $3$-colouring of $W$ depicted in \Cref{fig:3col}, no vertex has three monochromatic neighbours. 
\end{claimproof}

Suppose, to the contrary, that $g(W)$ contains an induced copy of $\overline{K_{3,3}+P_1}$ with vertex set $\{v_0, \ldots, v_6\}$, where $v_0$ is the universal vertex and $\{v_1,v_2,v_3\}$ and $\{v_4,v_5,v_6\}$ induce disjoint cliques. By \Cref{433c1}, the two copies of $K_4$ induced by $\{v_0,v_1,v_2,v_3\}$ and $\{v_0,v_4,v_5,v_6\}$ must both be monochromatic. Hence, $v_1,\ldots,v_6$ must be of the same colour. This implies that $v_1,\ldots,v_6$ must form a clique in $g(W)$, a contradiction.
\end{proof}

\begin{theorem}
\label{522t1}
The class of $(5P_1,\overline{K_{2,2}+P_1})$-free graphs has unbounded mim-width.
\end{theorem}

\begin{proof}
Let $W$ be an elementary $2n\times 2n$ wall and consider its proper $4$-colouring depicted in \Cref{fig:4col}. We add edges within each colour class to make them cliques. Let $h(W)$ be the graph obtained and let $\mathcal{W}_3 = \{h(W) : W \in \mathcal{W}\}$. By \Cref{walls} and \Cref{bipcompl}, $\mathcal{W}_3$ has unbounded mim-width. Clearly, $h(W)$ is $5P_1$-free. It remains to show that $h(W)$ is $\overline{K_{2,2}+P_1}$-free.

\begin{claim}
\label{522c1}
Any copy of $K_3$ in $h(W)$ is monochromatic. 
\end{claim}

\begin{claimproof}[Proof of \Cref{522c1}] Let $u_1,u_2,u_3$ be the vertices of a copy of $K_3$. Firstly, there cannot be exactly two vertices in $\{u_1,u_2,u_3\}$ of the same colour, say colour $c$, since this would imply that in $W$ there is a vertex coloured different from $c$ which is adjacent to two vertices coloured $c$, contradicting the $4$-colouring of $W$ depicted in \Cref{fig:4col}. Moreover, the vertices in $\{u_1,u_2,u_3\}$ cannot be coloured with distinct colours, for otherwise these three vertices would induce a $K_3$ in $W$. 
\end{claimproof}

Similarly to \Cref{344t1,433t1}, it is now easy to see that $h(W)$ is $\overline{K_{2,2}+P_1}$-free.
\end{proof}

\subsection{Dichotomy}\label{dichotomy}

Combining the results of \Cref{boundededgeless,unboundededgeless}, we can finally show \Cref{edgeless}, which we restate for convenience.

\edgeless*

\begin{proof} If $r \geq 5$, the mim-width is unbounded by \Cref{522t1}. Suppose now that $r = 4$. If both $s$ and $t$ are at least $3$, the mim-width is unbounded by \Cref{433t1}, whereas if one of $s$ and $t$ is at most $2$, the mim-width is bounded and quickly computable by \Cref{42tt1}. Finally, suppose that $r = 3$. If both $s$ and $t$ are at least $4$, the mim-width is unbounded by \Cref{344t1}, whereas if one of $s$ and $t$ is at most $3$, the mim-width is bounded and quickly computable by \Cref{33tt1}.   
\end{proof}

\section{Mim-width of $(K_r,sP_1+tP_2+uP_3)$-free graphs}\label{unboundedcomplete}

In this section we address \Cref{o-3} and show \Cref{completefour,completefive}. Both results are obtained by identifying new $(K_r,sP_1+tP_2+uP_3)$-free classes of unbounded mim-width. 

\Cref{o-3} was formulated in \citep{BHMPP22} starting from \citep[Theorem~35]{BHMPP22}. We remark that there is a typo in the formulation of this statement. For completeness we provide the correct formulation, whose proof is essentially identical to that of \citep[Theorem~35]{BHMPP22}. 

\begin{theorem}[\citet{BHMPP22}]
  Let $H$ be a graph and let $r \geq 4$ be an integer. Let ${\cal S}$ be the class of graphs every component of which is either a subdivided claw or a path. 
  Then exactly one of the following holds:

  \begin{itemize}
    \item $H \ssi sP_1+P_5$ or $tP_2$, and the mim-width of the class of
      $(K_r, H)$-free graphs is bounded and quickly computable;
    \item $H \notin \mathcal{S}$, or $H \si K_{1,3}$, $P_2+P_4$,
      or $P_6$, and the mim-width of the class of $(K_r, H)$-free graphs
      is unbounded; or
    \item $H = sP_1 + tP_2 + uP_3$, where $u \geq 1$ and $t+u \geq 2$.
  \end{itemize}
\end{theorem}

\begin{proof} By \citep[Theorem~31-(i)]{BHMPP22}, if $H \notin \mathcal{S}$, then the mim-width of the class of $(K_r,H)$-free graphs is unbounded.
  So we may assume that $H$ is a forest of paths and subdivided claws.
  By \citep[Theorem~31-(iii)]{BHMPP22}, if $H$ contains a $K_{1,3}$, then the mim-width is again unbounded.
  So we may assume that $H$ is a linear forest.
  If $H \ssi sP_1 + P_5$ or $H \ssi tP_2$, then mim-width is bounded
  and quickly computable by parts~(xii) and~(xiv) of \citep[Theorem~30]{BHMPP22}.
  So we may assume that $H$ is a linear forest containing $P_2 + P_3$.
  By \citep[Theorem~31-(viii)]{BHMPP22}, we may also assume $H$ contains neither $P_2 + P_4$ nor $P_6$, otherwise the mim-width is again unbounded.
  It now follows that $H \ssi tP_2 + uP_3$ for some $u,t$ such that $u \geq 1$ and $t+u \geq 2$. Therefore, $H = sP_1 + tP_2 + uP_3$, for $u \geq 1$ and $t + u \geq 2$.
\end{proof}

\subsection{Unboundedness results}

Similarly to \Cref{unboundededgeless}, the unboundedness results for $(K_r,sP_1+tP_2+uP_3)$-free graphs in this section (\Cref{K5unbound} for $r=5$ and \Cref{K4unbounded} for $r = 4$) are obtained by applying \Cref{bipcompl}. However, in the case of \Cref{K4unbounded}, only certain types of edges are added inside each colour class; this is to avoid creating copies of $K_4$. We will also make use of the following two results.  
 
\begin{lemma}[\citet{Vat12}]\label{delete} Let $G$ be a graph and $v \in V(G)$. Then $\mimw(G) \geq \mimw(G - v)$.
\end{lemma}

\begin{lemma}[\citet{BHMPP22}]\label{subd} Let $G$ be a graph and let $G'$ be the graph obtained by $1$-subdividing an edge of $G$. Then $\mimw(G') \geq \mimw(G)$.
\end{lemma}

\begin{theorem}\label{K5unbound}
The class of $(K_5,P_3+P_2+P_1)$-free graphs has unbounded mim-width.
\end{theorem}

\begin{proof}
Consider first a $2n \times 2n$-grid $G_{2n}$ with vertex set $\{(i,j) : 1 \leq i,j \leq 2n\}$. Consider the set of vertices $S = \{(i,j) : i+j\equiv 1 \Mod{ 2}\}$ and the set of edges $T = \{(i,j)(i,j-1) : (i,j) \in S\}$. We define the graph $W_n$ as $W_n = (V(G_n),E(G_n)\setminus T)$. Since $W_n$ contains the elementary $n \times n$ wall as an induced subgraph, \Cref{walls} and \Cref{delete} imply that the class of graphs $\{W_n : n \geq 1\}$ has unbounded mim-width. Given $W_n$, we now consider the following partition of its vertices: 
\begin{align*}
A &= \{(i,j) : i+j\equiv 0 \Mod{2}, \ i \equiv 1 \Mod{3}\} \\
B &= \{(i,j) : i+j\equiv 0 \Mod{2}, \ i \equiv 2 \Mod{3}\} \\
C &= \{(i,j) : i+j\equiv 0 \Mod{2}, \ i \equiv 0 \Mod{3}\} \\
D &= \{(i,j) : i+j\equiv 1 \Mod{2}, \ i \equiv 1 \Mod{3}\} \\
E &= \{(i,j) : i+j\equiv 1 \Mod{2}, \ i \equiv 2 \Mod{3}\} \\
F &= \{(i,j) : i+j\equiv 1 \Mod{2}, \ i \equiv 0 \Mod{3}\}.
\end{align*}
We then colour in red the vertices in $A\cup B \cup C$, and in blue the vertices in $D \cup E \cup F$ (see \Cref{fig:K5unbounded}). This gives a proper $2$-colouring of $W_n$ and, in particular, each partition class defined above forms an independent set in $W_n$. Observe that each vertex is adjacent to at most one vertex from each partition class of the opposite colour. That is, each vertex in $A\cup B \cup C$ is adjacent to at most one vertex from each of $D$, $E$ and $F$, and each vertex in $D\cup E \cup F$ is adjacent to at most one vertex from each of $A$, $B$ and $C$.  

\begin{figure}[h!]
\centering
\includegraphics[scale=0.8]{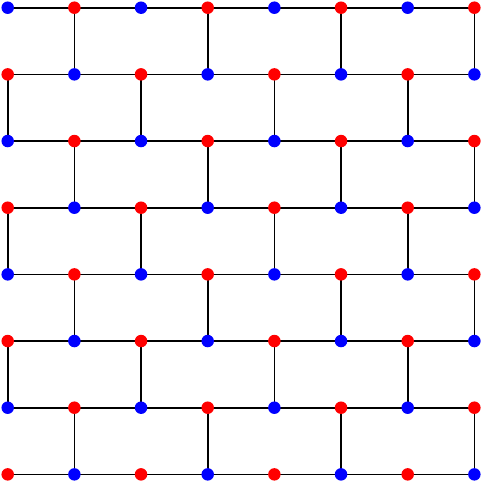}
\caption{The graph $W_4$ with the red-blue colouring as in the proof of \Cref{K5unbound}.}
\label{fig:K5unbounded}
\end{figure}

We now build the graph $W'_n$ by adding all edges between different partition classes of the same colour. That is, we make $A$, $B$, $C$ pairwise complete and $D$, $E$, $F$ pairwise complete. No other edges are added. In particular, $W'_n[A\cup B \cup C]$ and $W'_n[D\cup E\cup F]$ are complete tripartite graphs.

Applying \Cref{bipcompl} to the bipartition $(A\cup B\cup C, D \cup E \cup F)$ of $V(W_n)$, we obtain that $\mimw(W'_n) \geq \mimw(W_n)/2$. Hence, the class of graphs $\{W'_n : n \geq 1\}$ has unbounded mim-width. It is then enough to show that $W'_n$ is $K_5$-free and $(P_3+P_2+P_1)$-free.

\begin{claim}\label{K5wall} $W'_n$ is $K_5$-free.
\end{claim}

\begin{claimproof}[Proof of \Cref{K5wall}] Suppose, to the contrary, that $\{v_1,v_2,v_3,v_4,v_5\}$ induces a copy of $K_5$ in $W'_n$. Since each of $A,B,C,D,E,F$ is an independent set, the $v_i$'s belong to different partition classes. In particular, without loss of generality, $v_1,v_2,v_3$ are red and $v_4,v_5$ blue, or vice versa. Since no edges between red and blue vertices are added when constructing $W'_n$, we have that $\{v_1,v_2,v_3\}$ is complete to $\{v_4,v_5\}$ in $W_n$. But this contradicts the fact that in $W_n$ no two vertices have two common neighbours.
\end{claimproof}

\begin{claim}\label{K5P3} $W'_n$ is $(P_3+P_2+P_1)$-free. 
\end{claim}

\begin{claimproof}[Proof of \Cref{K5P3}] Suppose, to the contrary, that $\{v_1,\ldots,v_6\}$ induces a copy of $P_3+P_2+P_1$ in $W'_n$, where $W'_n[v_1,v_2,v_3] \cong P_3$ with $v_2$ adjacent to both $v_1$ and $v_3$, $\{v_4,v_5\}$ is anticomplete to $\{v_1,v_2,v_3\}$ and induces a copy of $P_2$, and $\{v_6\}$ is anticomplete to $\{v_1,\ldots,v_5\}$. Suppose, without loss of generality, that $v_2$ is red.

\textbf{Case 1:} Both $v_1$ and $v_3$ are blue. Since $v_4$ is adjacent to at most one vertex from each blue partition class,  we have that $v_1$ and $v_3$ belong to different blue partition classes. By construction, these partition classes are complete, contradicting the fact that $v_1$ is non-adjacent to $v_3$ in $W'_n$.

\textbf{Case 2:} At least one of $v_1$ and $v_3$ is red. Without loss of generality, $v_1$ is red. Since each partition class forms an independent set in $W'_n$, we have that $v_1$ does not belong to the class of $v_2$. But then $\{v_1,v_2\}$ dominates the red vertices and so $v_4,v_5,v_6$ are all blue. By a similar reasoning, $v_4,v_5,v_6$ all belong to the same blue partition class, or else there exists a vertex in $\{v_4,v_5,v_6\}$ dominating the remaining two. But each partition class is an independent set, contradicting the fact that $v_4$ is adjacent to $v_5$. 
\end{claimproof}
This concludes the proof of \Cref{K5unbound}.
\end{proof}

\begin{theorem}\label{K4unbounded}
The class of $(K_4,P_3+2P_2+P_1,2P_3+P_2)$-free graphs has unbounded mim-width.
\end{theorem}

\begin{proof} Let $W_n$ be the graph defined in the proof of \Cref{K5unbound}. Given $W_n$, we subdivide every edge $(i_1,j_1)(i_2,j_2)$ by adding a new vertex $(\frac{i_1+i_2}{2},\frac{j_1+j_2}{2})$.  We then multiply the coordinates of all vertices by $2$ (so, e.g., $(4,5.5)$ becomes $(8,11)$) and preserve the adjacencies between vertices in order to obtain a new graph $W'_n$. By \Cref{subd}, $\mimw(W'_n) \geq \mimw(W_n)$. We now define a partition of the vertices of $W'_n$ as follows (see \Cref{fig:K4unbounded}):
\begin{align*}
X &= \{(i,j) : i+j \equiv 2 \Mod{4}\} \\
Y &= \{(i,j) : i+j \equiv 0 \Mod{4}\} \\
A &= \{(i,j) : i+j \ \mbox{is odd},\ i \equiv 1 \Mod{3}\} \\
B &= \{(i,j) : i+j \ \mbox{is odd},\ i \equiv 2 \Mod{3}\} \\
C &= \{(i,j) : i+j \ \mbox{is odd},\ i \equiv 0 \Mod{3}\}
\end{align*}
Note that $X$ and $Y$ consist of the vertices of $W_n$, and $A$, $B$ and $C$ consist of the new vertices introduced after edge subdivisions. In particular, each partition class is an independent set. Moreover, $X$ is anticomplete to $Y$, and $A,B,C$ are pairwise anticomplete. Since $W'_n$ has no cycle of length $4$, each $x \in X$ and $y \in Y$ have at most one common neighbour in $A\cup B \cup C$. 

\begin{figure}[h!]
\centering
\includegraphics[scale=0.7]{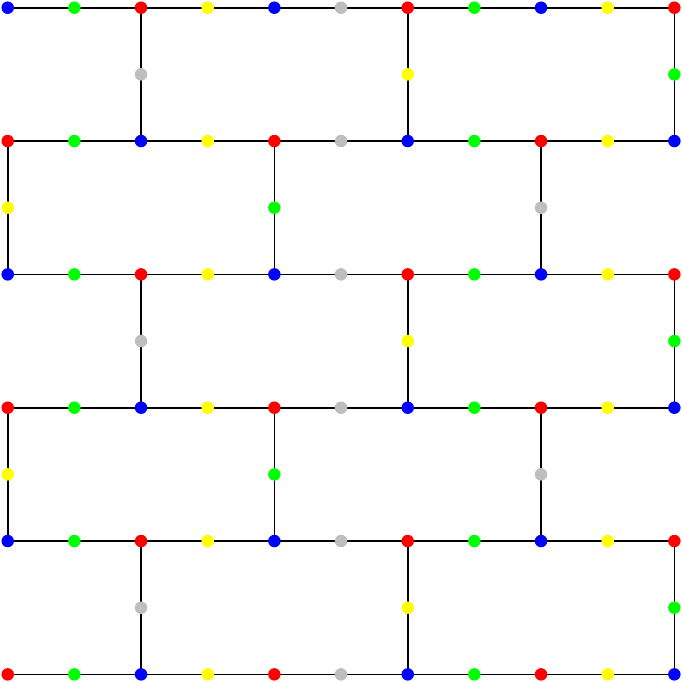}
\caption{The graph $W'_3$ in the proof of \Cref{K4unbounded}, together with a proper $5$-colouring: blue vertices correspond to $X$, red vertices to $Y$, grey vertices to $A$, yellow vertices to $B$ and green vertices to $C$.}
\label{fig:K4unbounded}
\end{figure}
\begin{observation}\label{obs1} Let $u_1 = (i_1,j_1)$ and $u_2 = (i_2,j_2)$ be two vertices belonging to the same partition class in $\{A, B, C\}$. The following hold: 
\begin{itemize}
\item $3$ divides $|i_1-i_2|$; 
\item If $i_1 = i_2$, then $2$ divides $|j_1-j_2|$. 
\end{itemize}
\end{observation}

\noindent We now proceed to the construction of the graph $W''_n$, obtained as follows. Start from $W'_n$ and
\begin{itemize}
\item Add all edges between $X$ and $Y$; \item For each pair of distinct sets $R$ and $S$ in $\{A, B, C\}$ and $r=(i_r,j_r) \in R$ and $s=(i_s,j_s) \in S$, add the edge $rs$, unless $j_r=j_s$ and $|i_r-i_s| = 2$, that is, unless $r$ and $s$ are the ``right neighbour'' and the ``left neighbour'' of a vertex in $X \cup Y$.
\end{itemize}
The edges left out in the second step above avoid the creation of copies of $K_4$, as will be shown shortly.

Since $(X \cup Y, A\cup B\cup C)$ is a bipartition of $V(W''_n)$, \Cref{bipcompl} implies that $\mimw(W''_n) \geq \mimw(W'_n)/2$. Hence the class of graphs $\{W''_n : n \geq 1\}$ has unbounded mim-width. It is then enough to show that $W''_n$ does not contain any graph in $\{K_4,P_3+2P_2+P_1,2P_3+P_2\}$ as an induced subgraph. This will be done in a series of claims. 

\begin{claim}\label{K4free} $W''_n$ is $K_4$-free. 
\end{claim}

\begin{claimproof}[Proof of \Cref{K4free}] Suppose, to the contrary, that $\{v_1,v_2,v_3,v_4\}$ induces a copy of $K_4$ in $W''_n$. Since each of $X,Y,A,B$ and $C$ is an independent set, the four vertices belong to four different partition classes. 

Suppose first that exactly one of $\{v_1,v_2,v_3,v_4\}$ belongs to $X \cup Y$. Without loss of generality, $v_1 \in X$ and $v_2,v_3,v_4 \in A \cup B \cup C$. Since no edges between $X$ and $A \cup B \cup C$ are added to $E(W'_n)$ in order to build $W''_n$, the vertices $v_2,v_3,v_4$ are adjacent to $v_1$ in $W'_n$. Suppose that $v_1 = (i,j)$. Then, up to relabelling, we must have that $v_2 = (i-1,j),v_3=(i+1,j)$ and $v_4 = (i,j \pm 1)$. In other words, $v_2$ and $v_3$ are the left neighbour and right neighbour of $v_1$, respectively. But then, by construction, $v_2v_3 \notin E(W''_n)$, a contradiction.

Suppose finally that exactly two vertices of $\{v_1,v_2,v_3,v_4\}$ belong to $X \cup Y$. Without loss of generality, $v_1,v_2 \in X \cup Y$, and $v_3,v_4 \in A \cup B \cup C$. Since no edges between $X \cup Y$ and $A \cup B \cup C$ are added to $E(W'_n)$ in order to build $W''_n$, both $v_1$ and $v_2$ are adjacent to $v_3$ and $v_4$ in $W'_n$, contradicting the fact that $W'_n$ does not contain any cycle of length $4$.
\end{claimproof}

\begin{claim}\label{obs2} Let $u_1,u_2$ be two distinct vertices from the same partition class in $\{A,B,C\}$. Let $u_3$ be a vertex from a partition class in $\{A,B,C\}$ different from that of $u_1$ and $u_2$. Then $u_3$ is adjacent to at least one of $u_1$ and $u_2$.
\end{claim}

\begin{claimproof}[Proof of \Cref{obs2}]
Let $u_1 = (i_1,j_1)$, $u_2 =(i_2,j_2)$ and $u_3 = (i_3,j_3)$. Suppose, to the contrary, that $u_3$ is non-adjacent to both $u_1$ and $u_2$. By construction of $W''_n$, this implies that $j_1 = j_3 = j_2$ and $|i_1-i_3| = 2 = |i_2-i_3|$. Since $u_1$ and $u_2$ are distinct, $i_1 \neq i_2$, which implies that $|i_1-i_2| = 4$, contradicting the first part of \Cref{obs1}.
\end{claimproof}

We now prove that $W''_n$ is $(P_3+2P_2+P_1)$-free and $(2P_3+2P_2)$-free. The following result will be used as the backbone of both proofs. 

\begin{claim}\label{322}
Suppose that $\{v_1,\ldots,v_7\}$ induces a copy of $P_3+2P_2$, where $v_2$ is adjacent to $v_1$ and $v_3$, $v_4$ is adjacent to $v_5$, $v_6$ is adjacent to $v_7$ and no other edges are present in $W''_n[\{v_1,\ldots,v_7\}]$. Then the following hold:
\begin{enumerate}
\item At least one of $v_1$ and $v_3$ belongs to $X\cup Y$;
\item $v_2 \in X\cup Y$.
\end{enumerate}
\end{claim}

\begin{claimproof}[Proof of \Cref{322}] We first show that at least one of $v_1$ and $v_3$ belongs to $X\cup Y$. Suppose, to the contrary, that both $v_1$ and $v_3$ belong to $A\cup B\cup C$. Since $A$, $B$ and $C$ are pairwise disjoint, $v_1,v_3 \in S \cup T$ for some distinct $S, T \in \{A, B, C\}$.

Observe that at least two of $v_4, v_5, v_6, v_7$, say $v_i$ and $v_j$, belong to $A \cup B \cup C$, or else at least three vertices among $v_4, v_5, v_6, v_7$ belong to $X\cup Y$ and so $W''_n[X\cup Y]$ contains a copy of $P_2+P_1$, contradicting the fact that $W''_n[X\cup Y]$ is a complete bipartite graph.

Observe now that, by \Cref{obs2} and the previous paragraph, $v_1$ and $v_3$ belong the same partition class. Without loss of generality, $v_1, v_3 \in S$. Since $S$ is an independent set, $v_2 \notin S$, and since each vertex in $X \cup Y$ has at most one neighbour in each of $A$, $B$ and $C$, we have that $v_2 \notin X \cup Y$.  Moreover, by \Cref{obs2}, $v_i$ and $v_j$ both belong to $S$. But this contradicts \Cref{obs2}, as $v_2 \in (A\cup B\cup C)\setminus S$. 

We finally show that $v_2 \in X \cup Y$. Suppose, to the contrary, that $v_2 \in R$, for some $R \in \{A, B, C\}$. Since $R$ is an independent set, $v_1,v_3 \not \in R$. In view of part 1, we distinguish two cases, according to which partition classes $v_1$ and $v_3$ belong. Let $S$ and $T$ be the two distinct partition classes in $\{A, B, C\} \setminus R$. 

\noindent \textbf{Case 1:} $v_1$ and $v_3$ both belong to $X\cup Y$. 

Since each vertex in $R$ is adjacent to at most one vertex in $X$ and at most one vertex in $Y$, one of $v_1$ and $v_3$ belongs to $X$ and the other to $Y$, contradicting the fact that $X$ is complete to $Y$.

\noindent \textbf{Case 2:} One of $v_1$ and $v_3$ belongs to $X \cup Y$ and the other to $S \cup T$.

Without loss of generality, $v_1 \in S$ and $v_3 \in X$. Since $X$ is complete to $Y$,  $v_4,v_5,v_6,v_7 \not \in Y$. Since $X$ is an independent set, at most one of $v_4$ and $v_5$ belongs to $X$ and at most one of $v_6$ and $v_7$ belongs to $X$. Without loss of generality, $v_4,v_6 \in A\cup B \cup C$. If both $v_4$ and $v_6$ belong to $R$, then $v_1 \in S$ is non-adjacent to both $v_4,v_6 \in R$, contradicting \Cref{obs2}. If exactly one of $v_4$ and $v_6$ belongs to $R$, say without loss of generality $v_4 \in R$ and $v_6 \notin R$, then $v_6 \in (A \cup B \cup C)\setminus R$ is non-adjacent to $v_2 \in R$ and $v_4 \in R$, contradicting \Cref{obs2}. Therefore, none of $v_4$ and $v_6$ belongs to $R$. If $v_4$ and $v_6$ belong to the same partition class in $(A \cup B \cup C)\setminus R$, then $v_2 \in R$ being non-adjacent to both of them contradicts \Cref{obs2}. Finally, if $v_4$ and $v_6$ belong to different partition classes in $(A \cup B \cup C)\setminus R$, then one of them belongs to the partition class $S$ of $v_1$, say without loss of generality $v_4 \in S$. But then $v_6$ being non-adjacent to both $v_1$ and $v_4$ contradicts \Cref{obs2}.
\end{claimproof}

\begin{claim}\label{P32P2} $W''_n$ is $(P_3+2P_2+P_1)$-free.
\end{claim}
    
\begin{claimproof}[Proof of \Cref{P32P2}] Suppose, to the contrary, that $\{v_1,\ldots,v_8\}$ induces a copy of $P_3+2P_2+P_1$, where $v_2$ is adjacent to $v_1$ and $v_3$, $v_4$ is adjacent to $v_5$, $v_6$ is adjacent to $v_7$ and no other edges are present in $W''_n[\{v_1,\ldots,v_8\}]$ (hence $v_8$ is the isolated vertex). By \Cref{322}, $v_2 \in X\cup Y$ and at least one of $v_1$ and $v_3$ belongs to $X\cup Y$. Without loss of generality, $v_2 \in X$ and $v_1 \in X \cup Y$. Since $X$ is an independent set, $v_1 \in Y$. Since $X$ is complete to $Y$, we have that $\{v_1,v_2\}$ dominates $X \cup Y$ and so $\{v_4,\ldots,v_8\} \subseteq A\cup B \cup C$. By the pigeonhole principle, there exists two vertices among $v_4,v_5,v_6,v_7$ that belong to the same partition class in $\{A,B,C\}$. Since these classes form independent sets, the two vertices are non-adjacent. Without loss of generality, $v_4,v_6 \in R$ for some $R \in \{A, B, C\}$. If $v_8 \in (A \cup B \cup C)\setminus R$, then $v_8$ is non-adjacent to both $v_4,v_6 \in R$, contradicting \Cref{obs2}. Therefore, $v_8 \in R$. Since $R$ is an independent set, $v_4 \in R$ implies that $v_5 \in (A \cup B \cup C)\setminus R$ and $v_5$ is non-adjacent to both $v_6,v_8 \in R$, contradicting \Cref{obs2}.
\end{claimproof}
    
\begin{claim}\label{2P3P2} $W''_n$ is $(2P_3+P_2)$-free.
\end{claim}
    
\begin{claimproof}[Proof of \Cref{2P3P2}] Suppose, to the contrary, that $\{v_1,\ldots,v_8\}$ induces a copy of $2P_3+P_2$, where $v_2$ is adjacent to $v_1$ and $v_3$, $v_4$ is adjacent to $v_5$, $v_7$ is adjacent to $v_6$ and $v_8$, and no other edges are present in $W''_n[\{v_1,\ldots,v_8\}]$. By \Cref{322}, $v_2 \in X\cup Y$ and at least one of $v_1$ and $v_3$ belongs to $X\cup Y$. Without loss of generality, $v_2 \in X$ and $v_1 \in X \cup Y$. As in the proof of \Cref{P32P2}, $\{v_1,v_2\}$ dominates $X \cup Y$ and so $\{v_4,\ldots,v_8\} \in A\cup B \cup C$. 
     
Suppose first that at least two vertices among $v_6$, $v_7$ and $v_8$ belong to the same partition class in $\{A,B,C\}$. These two vertices are non-adjacent, as $A,B,C$ are independent sets, and so they must be $v_6$ and $v_8$. Without loss of generality, $v_6,v_8 \in R$ for some $R \in \{A,B,C\}$. Similarly, at least one of $v_4$ and $v_5$ does not belong to $R$, say $v_4 \in (A \cup B \cup C)\setminus R$. Then $v_4$ is non-adjacent to both $v_6,v_8 \in R$, contradicting \Cref{obs2}.
     
Therefore, $v_6$, $v_7$ and $v_8$ belong to distinct partition classes in $\{A, B, C\}$. By \Cref{obs2}, none of $v_4$ and $v_5$ belongs to the partition class of either $v_6$ or $v_8$. But then $v_4$ and $v_5$ both belong to the partition class of $v_7$, contradicting the fact that every class is an independent set. 
\end{claimproof}
This concludes the proof of \Cref{K4unbounded}.
\end{proof}

\subsection{Summary results}

With the aid of \Cref{K5unbound,K4unbounded}, we can finally show \Cref{completefour,completefive}, which we restate for convenience.

\completefive*

\begin{proof} By \Cref{K5unbound}, if $H$ contains $P_3+P_2+P_1$, then the mim-width of the class of $(K_r, H)$-free graphs is unbounded. So we may assume that $u \leq 2$. If $u = 0$, then the mim-width is bounded by \citep[Theorem~30-(xiv)]{BHMPP22}. If $u = 1$, then the mim-width is unbounded for $t \geq 2$ and $s \geq 0$ or $t = 1$ and $s \geq 1$ (\Cref{K5unbound}), and bounded for $t=0$ (\citep[Theorem~30-(xii)]{BHMPP22}). This leaves open the case $H = P_3 + P_2$. Finally, if $u = 2$, then the mim-width is unbounded if one of $t$ and $s$ is at least $1$. This leaves open the case $H = 2P_3$.  
\end{proof}

\completefour*

\begin{proof} By \Cref{K4unbounded}, if $H$ contains $P_3+2P_2+P_1$ or $2P_3 + P_2$, then the mim-width of the class of $(K_r, H)$-free graphs is unbounded. So we may assume that $u \leq 2$. If $u = 0$, then the mim-width is bounded by \citep[Theorem~30-(xiv)]{BHMPP22}. If $u = 1$, then the mim-width is bounded for $t=0$ (\citep[Theorem~30-(xii)]{BHMPP22}), and unbounded for $t \geq 2$ and $s \geq 1$ or $t \geq 3$ and $s \geq 0$ (\Cref{K4unbounded}). This leaves open the cases $H = P_3 + 2P_2$ and $P_3+P_2+sP_1$. Finally, if $u = 2$, then the mim-width is unbounded for $t \geq 1$. This leaves open the case $H = 2P_3 + sP_1$.  
\end{proof}

\section{Concluding remarks and open problems}

In view of \Cref{reduction}, we believe that the main open problem related to algorithmic applications of sim-width is whether \textsc{Independent Set} is polynomial-time solvable for graph classes whose sim-width is bounded and quickly computable (this was first formulated in \citep{KKST17}). We highlight a possible connection. In \citep{DMS21}, it is asked whether there exists a $(\tw, \omega)$-bounded graph class for which \textsc{Independent Set} is $\mathsf{NP}$-hard. In view of these two open problems, it would be interesting to determine whether every $(\tw, \omega)$-bounded graph class has bounded sim-width (the converse does not hold, as mentioned in \Cref{ourresults}).

\citet{CH06} showed that \textsc{Independent $\mathcal{H}$-packing} is polynomial-time solvable for weakly chordal graphs, a superclass of chordal graphs, and for AT-free graphs, a superclass of co-comparability graphs. Both chordal graphs and co-comparability graphs have sim-width at most $1$ and in \citep{KKST17} it is asked whether weakly chordal graphs and AT-free graphs have bounded sim-width. We believe that \Cref{reduction} also provides strong motivation for studying the sim-width of weakly chordal and AT-free graphs. 

Finally, we conclude by asking to classify the mim-width for the remaining open cases in \Cref{completefour,completefive}. A particularly interesting open case is the mim-width of $(K_r, 2P_3)$-free graphs, for $r \geq 5$. In view of \Cref{listsim}, this is related to the open problem in \citep{HLS21} of whether there exists $k \in \mathbb{N}$ for which \textsc{List $k$-Colouring} restricted to $uP_3$-free graphs is $\mathsf{NP}$-hard for some $u \in \mathbb{N}$.

\section*{Acknowledgement}

We thank an anonymous referee for valuable comments.

\bibliographystyle{plainnat}
\bibliography{references}

\begin{thebibliography}{41}
\providecommand{\natexlab}[1]{#1}
\providecommand{\url}[1]{\texttt{#1}}
\expandafter\ifx\csname urlstyle\endcsname\relax
  \providecommand{\doi}[1]{doi: #1}\else
  \providecommand{\doi}{doi: \begingroup \urlstyle{rm}\Url}\fi

\bibitem[Belmonte and Vatshelle(2013)]{BV13}
Rémy Belmonte and Martin Vatshelle.
\newblock Graph classes with structured neighborhoods and algorithmic
  applications.
\newblock \emph{Theoretical Computer Science}, 511:\penalty0 54--65, 2013.

\bibitem[Bergougnoux and Kant{\'{e}}(2021)]{BK19}
Benjamin Bergougnoux and Mamadou~Moustapha Kant{\'{e}}.
\newblock More applications of the d-neighbor equivalence: {A}cyclicity and
  connectivity constraints.
\newblock \emph{{SIAM} Journal on Discrete Mathematics}, 35\penalty0
  (3):\penalty0 1881--1926, 2021.

\bibitem[Bergougnoux et~al.(2023)Bergougnoux, Dreier, and Jaffke]{BDJ22}
Benjamin Bergougnoux, Jan Dreier, and Lars Jaffke.
\newblock A logic-based algorithmic meta-theorem for mim-width.
\newblock In Nikhil Bansal and Viswanath Nagarajan, editors, \emph{Proceedings
  of the 2023 Annual ACM-SIAM Symposium on Discrete Algorithms (SODA)}, pages
  3282--3304. Society for Industrial and Applied Mathematics, 2023.

\bibitem[Bodlaender(1996)]{Bod96}
Hans~L. Bodlaender.
\newblock A linear-time algorithm for finding tree-decompositions of small
  treewidth.
\newblock \emph{SIAM Journal on Computing}, 25\penalty0 (6):\penalty0
  1305--1317, 1996.

\bibitem[Brettell et~al.(2020)Brettell, Horsfield, Munaro, Paesani, and
  Paulusma]{BHMPP1}
Nick Brettell, Jake Horsfield, Andrea Munaro, Giacomo Paesani, and Dani{\"{e}}l
  Paulusma.
\newblock Bounding the mim-width of hereditary graph classes.
\newblock In Yixin Cao and Marcin Pilipczuk, editors, \emph{15th International
  Symposium on Parameterized and Exact Computation, {IPEC} 2020}, volume 180 of
  \emph{LIPIcs}, pages 6:1--6:18. Schloss Dagstuhl - Leibniz-Zentrum f{\"{u}}r
  Informatik, 2020.

\bibitem[Brettell et~al.(2022{\natexlab{a}})Brettell, Horsfield, Munaro,
  Paesani, and Paulusma]{BHMPP22}
Nick Brettell, Jake Horsfield, Andrea Munaro, Giacomo Paesani, and Dani{\"{e}}l
  Paulusma.
\newblock Bounding the mim-width of hereditary graph classes.
\newblock \emph{Journal of Graph Theory}, 99\penalty0 (1):\penalty0 117--151,
  2022{\natexlab{a}}.

\bibitem[Brettell et~al.(2022{\natexlab{b}})Brettell, Horsfield, Munaro, and
  Paulusma]{BHMP22}
Nick Brettell, Jake Horsfield, Andrea Munaro, and Dani{\"{e}}l Paulusma.
\newblock List $k$-colouring {$P_t$}-free graphs: {A} mim-width perspective.
\newblock \emph{Information Processing Letters}, 173, 2022{\natexlab{b}}.
\newblock article number 106168.

\bibitem[Bui-Xuan et~al.(2011)Bui-Xuan, Telle, and Vatshelle]{BTV11}
Binh-Minh Bui-Xuan, Jan~Arne Telle, and Martin Vatshelle.
\newblock Boolean-width of graphs.
\newblock \emph{Theoretical Computer Science}, 412\penalty0 (39):\penalty0
  5187--5204, 2011.

\bibitem[Bui-Xuan et~al.(2013)Bui-Xuan, Telle, and Vatshelle]{BTV13}
Binh-Minh Bui-Xuan, Jan~Arne Telle, and Martin Vatshelle.
\newblock Fast dynamic programming for locally checkable vertex subset and
  vertex partitioning problems.
\newblock \emph{Theoretical Computer Science}, 511:\penalty0 66--76, 2013.

\bibitem[Cameron and Hell(2006)]{CH06}
Kathie Cameron and Pavol Hell.
\newblock Independent packings in structured graphs.
\newblock \emph{Mathematical Programming}, 105\penalty0 (2-3):\penalty0
  201--213, 2006.

\bibitem[Chaplick and Zeman(2017)]{CZ17}
Steven Chaplick and Peter Zeman.
\newblock Combinatorial problems on {$H$}-graphs.
\newblock \emph{Electronic Notes in Discrete Mathematics}, 61:\penalty0
  223--229, 2017.
\newblock The European Conference on Combinatorics, Graph Theory and
  Applications (EUROCOMB'17).

\bibitem[Chudnovsky et~al.(2020)Chudnovsky, Spirkl, and Zhong]{CSZ20}
Maria Chudnovsky, Sophie Spirkl, and Mingxian Zhong.
\newblock List $3$-coloring {$P_t$}-free graphs with no induced $1$-subdivision
  of {$K_{1,s}$}.
\newblock \emph{Discrete Mathematics}, 343\penalty0 (11):\penalty0 112086,
  2020.

\bibitem[Chudnovsky et~al.(2021)Chudnovsky, King, Pilipczuk, Rz\c{a}\.{z}ewski,
  and Spirkl]{CKPRS21}
Maria Chudnovsky, Jason King, Micha\l{} Pilipczuk, Pawe\l{} Rz\c{a}\.{z}ewski,
  and Sophie Spirkl.
\newblock Finding large {$H$}-colorable subgraphs in hereditary graph classes.
\newblock \emph{SIAM Journal on Discrete Mathematics}, 35\penalty0
  (4):\penalty0 2357--2386, 2021.

\bibitem[Courcelle et~al.(2000)Courcelle, Makowsky, and Rotics]{CMR00}
B.~Courcelle, J.~A. Makowsky, and U.~Rotics.
\newblock Linear time solvable optimization problems on graphs of bounded
  clique-width.
\newblock \emph{Theory of Computing Systems}, 33\penalty0 (2):\penalty0
  125--150, 2000.

\bibitem[Courcelle(1990)]{Cou90}
Bruno Courcelle.
\newblock The monadic second-order logic of graphs. {I}. {R}ecognizable sets of
  finite graphs.
\newblock \emph{Information and Computation}, 85\penalty0 (1):\penalty0 12--75,
  1990.

\bibitem[Courcelle and Olariu(2000)]{CO00}
Bruno Courcelle and Stephan Olariu.
\newblock Upper bounds to the clique width of graphs.
\newblock \emph{Discrete Applied Mathematics}, 101\penalty0 (1–3):\penalty0
  77--114, 2000.

\bibitem[Couturier et~al.(2015)Couturier, Golovach, Kratsch, and
  Paulusma]{CGKP15}
Jean{-}Fran{\c{c}}ois Couturier, Petr~A. Golovach, Dieter Kratsch, and
  Dani{\"{e}}l Paulusma.
\newblock List coloring in the absence of a linear forest.
\newblock \emph{Algorithmica}, 71\penalty0 (1):\penalty0 21--35, 2015.

\bibitem[Dabrowski et~al.(2019)Dabrowski, Johnson, and Paulusma]{DJP19}
Konrad~K. Dabrowski, Matthew Johnson, and Daniël Paulusma.
\newblock \emph{Clique-width for hereditary graph classes}, page 1–56.
\newblock London Mathematical Society Lecture Note Series. Cambridge University
  Press, 2019.

\bibitem[Dallard et~al.(2022)Dallard, Milani\v{c}, and \v{S}torgel]{DMS21}
Clément Dallard, Martin Milani\v{c}, and Kenny \v{S}torgel.
\newblock Treewidth versus clique number. {II}. {T}ree-independence number.
\newblock \emph{CoRR}, abs/2111.04543, 2022.
\newblock URL \url{https://arxiv.org/abs/2111.04543}.

\bibitem[Galby et~al.(2020)Galby, Munaro, and Ries]{GMR20}
Esther Galby, Andrea Munaro, and Bernard Ries.
\newblock Semitotal domination: New hardness results and a polynomial-time
  algorithm for graphs of bounded mim-width.
\newblock \emph{Theoretical Computer Science}, 814:\penalty0 28--48, 2020.

\bibitem[Garey et~al.(1980)Garey, Johnson, Miller, and Papadimitriou]{GJMP80}
M.~R. Garey, D.~S. Johnson, G.~L. Miller, and C.~H. Papadimitriou.
\newblock The complexity of coloring circular arcs and chords.
\newblock \emph{SIAM Journal on Algebraic Discrete Methods}, 1\penalty0
  (2):\penalty0 216--227, 1980.

\bibitem[Golovach et~al.(2014)Golovach, Paulusma, and Song]{GPS14b}
Petr~A. Golovach, Daniël Paulusma, and Jian Song.
\newblock Coloring graphs without short cycles and long induced paths.
\newblock \emph{Discrete Applied Mathematics}, 167:\penalty0 107--120, 2014.

\bibitem[Golovach et~al.(2017)Golovach, Johnson, Paulusma, and Song]{GJPS17}
Petr~A. Golovach, Matthew Johnson, Daniël Paulusma, and Jian Song.
\newblock A survey on the computational complexity of coloring graphs with
  forbidden subgraphs.
\newblock \emph{Journal of Graph Theory}, 84\penalty0 (4):\penalty0 331--363,
  2017.

\bibitem[Hajebi et~al.(2022)Hajebi, Li, and Spirkl]{HLS21}
Sepehr Hajebi, Yanjia Li, and Sophie Spirkl.
\newblock Complexity dichotomy for {L}ist-5-{C}oloring with a forbidden induced
  subgraph.
\newblock \emph{SIAM Journal on Discrete Mathematics}, 36\penalty0
  (3):\penalty0 2004--2027, 2022.

\bibitem[Ho\`ang et~al.(2010)Ho\`ang, Kami\'nski, Lozin, Sawada, and
  Shu]{HKLSS10}
Ch\'inh~T. Ho\`ang, Marcin Kami\'nski, Vadim Lozin, Joe Sawada, and Xiao Shu.
\newblock Deciding $k$-{C}olorability of {$P_5$}-free graphs in polynomial
  time.
\newblock \emph{Algorithmica}, 57:\penalty0 74--81, 2010.

\bibitem[il~Oum and Seymour(2006)]{OS06}
Sang il~Oum and Paul Seymour.
\newblock Approximating clique-width and branch-width.
\newblock \emph{Journal of Combinatorial Theory, Series B}, 96\penalty0
  (4):\penalty0 514--528, 2006.

\bibitem[Jaffke(2020)]{Jaf20}
Lars Jaffke.
\newblock \emph{Bounded Width Graph Classes in Parameterized Algorithms}.
\newblock PhD thesis, University of Bergen, 2020.

\bibitem[Jaffke et~al.(2019)Jaffke, Kwon, Str{\o}mme, and Telle]{JKST19}
Lars Jaffke, O{-}joung Kwon, Torstein J.~F. Str{\o}mme, and Jan~Arne Telle.
\newblock Mim-width {III.} {G}raph powers and generalized distance domination
  problems.
\newblock \emph{Theoretical Computer Science}, 796:\penalty0 216--236, 2019.

\bibitem[Jaffke et~al.(2020{\natexlab{a}})Jaffke, Kwon, and Telle]{JKT}
Lars Jaffke, O{-}joung Kwon, and Jan~Arne Telle.
\newblock Mim-width {I}. {I}nduced path problems.
\newblock \emph{Discrete Applied Mathematics}, 278:\penalty0 153--168,
  2020{\natexlab{a}}.

\bibitem[Jaffke et~al.(2020{\natexlab{b}})Jaffke, Kwon, and Telle]{JKT20}
Lars Jaffke, O{-}joung Kwon, and Jan~Arne Telle.
\newblock Mim-width {II. The Feedback Vertex Set} problem.
\newblock \emph{Algorithmica}, 82:\penalty0 118--145, 2020{\natexlab{b}}.

\bibitem[Jeong et~al.(2018)Jeong, Sæther, and Telle]{JST18}
Jisu Jeong, Sigve~Hortemo Sæther, and Jan~Arne Telle.
\newblock Maximum matching width: New characterizations and a fast algorithm
  for dominating set.
\newblock \emph{Discrete Applied Mathematics}, 248:\penalty0 114--124, 2018.

\bibitem[Kang et~al.(2017)Kang, Kwon, Str{\o}mme, and Telle]{KKST17}
Dong~Yeap Kang, O{-}joung Kwon, Torstein~J.F. Str{\o}mme, and Jan~Arne Telle.
\newblock A width parameter useful for chordal and co-comparability graphs.
\newblock \emph{Theoretical Computer Science}, 704:\penalty0 1--17, 2017.

\bibitem[Klimosov{\'{a}} et~al.(2020)Klimosov{\'{a}}, Mal{\'{\i}}k,
  Masar{\'{\i}}k, Novotn{\'{a}}, Paulusma, and Sl{\'{i}}vov{\'{a}}]{KMMNPS20}
Tereza Klimosov{\'{a}}, Josef Mal{\'{\i}}k, Tom{\'{a}}s Masar{\'{\i}}k, Jana
  Novotn{\'{a}}, Dani{\"{e}}l Paulusma, and Veronika Sl{\'{i}}vov{\'{a}}.
\newblock Colouring {$P_r + P_s$}-free graphs.
\newblock \emph{Algorithmica}, 82\penalty0 (7):\penalty0 1833--1858, 2020.

\bibitem[Kwon(2020)]{Kw20}
O{-}joung Kwon.
\newblock Personal communication, 2020.

\bibitem[Lozin and Razgon(2022)]{LR22}
Vadim Lozin and Igor Razgon.
\newblock Tree-width dichotomy.
\newblock \emph{European Journal of Combinatorics}, 103:\penalty0 103517, 2022.

\bibitem[Mengel(2018)]{Men17}
Stefan Mengel.
\newblock Lower bounds on the mim-width of some graph classes.
\newblock \emph{Discrete Applied Mathematics}, 248:\penalty0 28--32, 2018.

\bibitem[\"{O}jvind Johansson(1998)]{Jo98}
\"{O}jvind Johansson.
\newblock Clique-decomposition, {NLC}-decomposition, and modular decomposition
  - relationships and results for random graphs.
\newblock \emph{Congressus Numerantium}, 132:\penalty0 39--60, 1998.

\bibitem[Rao(2008)]{Ra08}
Michaël Rao.
\newblock Clique-width of graphs defined by one-vertex extensions.
\newblock \emph{Discrete Mathematics}, 308\penalty0 (24):\penalty0 6157--6165,
  2008.

\bibitem[Robertson and Seymour(1991)]{RS91}
Neil Robertson and P.D Seymour.
\newblock Graph minors. {X}. {O}bstructions to tree-decomposition.
\newblock \emph{Journal of Combinatorial Theory, Series B}, 52\penalty0
  (2):\penalty0 153--190, 1991.

\bibitem[S{\ae}ther and Vatshelle(2016)]{SV16}
Sigve~Hortemo S{\ae}ther and Martin Vatshelle.
\newblock Hardness of computing width parameters based on branch decompositions
  over the vertex set.
\newblock \emph{Theoretical Computer Science}, 615:\penalty0 120--125, 2016.

\bibitem[Vatshelle(2012)]{Vat12}
Martin Vatshelle.
\newblock \emph{New Width Parameters of Graphs}.
\newblock PhD thesis, University of Bergen, 2012.

\end{thebibliography}

\end{document}